\documentclass[11pt]{article}

\usepackage[T1]{fontenc}
\usepackage[utf8]{inputenc}
\usepackage{microtype}
\usepackage{mathpazo}

\usepackage[margin=1in]{geometry}

\usepackage{enumitem}
\setlist{nosep,topsep=0pt,itemsep=0pt,leftmargin=*}

\usepackage{algpseudocode}
\usepackage{graphicx}

\usepackage{maths}
\usepackage[ruled, linesnumbered]{algorithm2e}

\usepackage[
  colorlinks=true,
  linkcolor=RubineRed,
  citecolor=OliveGreen,
  urlcolor=Black,
  breaklinks=true,
  unicode
]{hyperref}
\usepackage{cleveref}

\title{Online Resource Allocation via Static Bundle Pricing}

\usepackage{authblk}
\usepackage{natbib}

\author[1,2]{Dimitris Fotakis}
\author[1,2]{Charalampos Platanos}
\author[1,2]{Thanos Tolias}

\affil[1]{National Technical University of Athens, Greece}
\affil[2]{Archimedes RU, Athena RC, Greece}

\affil[ ]{\texttt{\href{mailto:fotakis@cs.ntua.gr}{fotakis@cs.ntua.gr} \quad
\href{mailto:harrisplat@gmail.com}{harrisplat@gmail.com} \quad
\href{mailto:thanostolias@mail.ntua.gr}{thanostolias@mail.ntua.gr}}}

\date{}
\begin{document}

\maketitle

\begin{abstract}
Online Resource Allocation addresses the problem of efficiently allocating limited resources to buyers with incomplete knowledge of future requests. In our setting, buyers arrive sequentially requesting a set of items, each with a value drawn from a known distribution. We study the efficiency of static and anonymous bundle pricing in environments where the buyers' valuations exhibit strong complementarities. In such settings, standard item pricing  fails to leverage item multiplicities, while static bundle pricing mechanisms are only known for very restricted domains and their analysis relies on domain-specific arguments. 

We develop a unified bundle pricing framework for online resource allocation in three well-studied domains with complementarities: (i)~single-minded combinatorial auctions with maximum bundle size~$d$; (ii)~general single-minded combinatorial auctions; and (iii)~network routing, where each buyer aims to route a unit of flow from a source node $s$ to a target node $t$ in a capacitated network. Our approach yields static and anonymous bundle pricing mechanisms whose performance improves exponentially with item multiplicity. For the $d$-single-minded setting with minimum item multiplicity~$B$, we obtain an $O(d^{1/B})$-competitive mechanism. 
For general single-minded combinatorial auctions and online network routing, we obtain $O(m^{1/(B+1)})$-competitive mechanisms, where $m$ is the number of items.

We complement these results with information-theoretic lower bounds. We show that no online algorithm can achieve a competitive ratio better than $ \widetilde{\Omega}(m^{1/(B+2)})$ for single-minded combinatorial auctions and $ \widetilde{\Omega}(d^{1/(B+1)})$ for the $d$-single-minded setting. Our constructions exploit a deep connection to the classical extremal combinatorics problem of determining the maximum number of qualitatively independent partitions of a ground set.
\end{abstract}

\maketitle

\newpage
\section{Introduction}

Online Resource Allocation is a central problem in Economics and Operations Research. A seller faces a sequence of buyers, each demanding a set of resources, and must make irrevocable allocation decisions with limited information about future arrivals. In this work, we study online resource allocation in a Bayesian setting where the seller has access to the distribution of buyers' valuations, placing the problem within the general prophet-inequality framework.

\paragraph{General Background.}
The study of prophet inequalities originates in the classical single-item setting. \citet{krengel1977} (also crediting D.J.H.~Garling) and independently \citet{samuel1984comparison} established the foundational results comparing the value obtained by online stopping rules against the best possible (a.k.a. optimal) value obtained by a ``prophet'' who observes all realizations in advance. More recently \citet{haji1007} observed the connection of prophet inequalities to mechanism design, showing how prophet-style guarantees can be leveraged to design truthful posted-price mechanisms. Building on this connection, \citet{chawla2009} initiated the study of prophet inequalities for revenue maximization, further expanding the framework's applicability beyond the social welfare.

Following these foundational results, research on prophet inequalities expanded rapidly, branching into richer allocation environments and more complex valuation models. In particular, combinatorial auctions emerged as a central and intensively studied setting within the prophet-inequality framework. For submodular and XOS valuations, \citet{fgl-balanced14} significantly generalized the basic prophet inequality, introducing the elegant framework of balanced prices, and obtained a $2$-competitive posted-item-price mechanism. For subadditive valuations, \citet{duttingloglogm} obtained an $O(\log\log m)$-competitive item-pricing mechanism (where $m$ is the number of items), a guarantee recently improved by  \citet{Correa2023ACF}, who achieved a $6$-competitive prophet inequality. 
For MPH-$d$ valuations, which exhibit a degree of complementarity quantified by $d$, \citet{braun2022simplifiedprophetinequalitiescombinatorial} presented an item-pricing mechanism with competitive ratio $2d + \sqrt{2d(d-1)} - 1$. 

In this work, we focus on single-minded combinatorial auctions, which exhibit strong complementarities: buyers derive value only if they obtain their entire desired bundle. Single-minded combinatorial auctions capture the most extreme form of complementarities and arise naturally in many applications where partial allocation provides no value. Their simplicity makes them a fundamental testbed for understanding how complementarities interact with online decision-making and posted-price mechanisms. For a simple canonical example, we may refer to replenishing and pricing inventory in a grocery store. A shopper attempting to bake a cake gains value only if they obtain all necessary ingredients; eggs and milk provide zero utility if flour is missing.

Extending the grocery store analogy, many practical applications feature multiple copies of each item rather than unique goods. A growing line of work examines how the performance of pricing mechanisms improves with item multiplicity. \citet{multi-alaei} showed that in the single-item setting with capacity $B$, one can achieve $1+O(1/\sqrt{B})$ competitiveness. Subsequently, \citet{multi-jiang} and \citet{multi-suchi} gave matching upper and lower bounds for all values of~$B$. Recently, \citet{ghuge2025} studied online resource allocation with large capacities under stochastic arrivals using \emph{dynamic pricing} mechanisms (see also e.g., \cite{DevanurJSW19,ChawlaDHKMS17,KesselheimRTV18,GuptaM16} and their references for more previous work on achieving nearly optimal competitiveness for online resource allocation in the large-capacity regime). \citet{ghuge2025} proved that for capacity $B = \widetilde\Omega(1/\eps^6)$, a competitive ratio of $1+O(\eps)$ can be achieved using only a single sample from each distribution, highlighting that high capacities substantially mitigate uncertainty even with limited information.

\paragraph{Context and Motivation.}
We focus on \emph{static and anonymous} pricing mechanisms, in which the seller commits to a menu of prices before any buyer arrives and does not adjust them over time. Although such mechanisms are simple, transparent and widely deployed in practice, they face significant limitations in the presence of complementarities. \citet{fgl-balanced14} 
proved that when buyers' valuations exhibit complementarities, no item-pricing mechanism can be less than $\Omega(m)$-competitive. When parameterizing by the maximum bundle size $d$, \citet{correa-d} obtained a tight $(d+1)$-competitive guarantee.  

Motivated by these challenges, \citet{chawla18} initiated the study of bundle pricing in settings with strong complementarities, in the context of online allocation of intervals on the line and paths in trees. Their results demonstrated that richer pricing menus can circumvent the barriers inherent to item pricing. For allocating intervals, where the length of the longest interval requested is $d$ and items are available in $B$ copies, \citet{chawla18} gave an $O\left(\frac{\log d}{B \log\log d}\right)$-competitive static, anonymous bundle-pricing scheme and a matching lower bound on the performance of any online algorithm. For the more challenging setting of allocating paths in a tree, they presented a lower bound of $\Omega(\sqrt{d/\log d})$ on the competitiveness of any online algorithm and a static, anonymous bundle pricing scheme with competitive ratio $O((\log v_{\max}) / B)$, where $v_{\max}$ is the maximum valuation. At a conceptual level, their results exploit the interval (resp. the tree) structure in order to reduce the setting with complementarities to the single item setting at a loss in the competitive ratio that scales with $\log d$ (resp. $\log v_{\max}$). Exploiting the restricted combinatorial structure of intervals, \citet{chawla18} improve exponentially on the best possible competitiveness of online allocations 
and show that the competitive ratio improves inversely with item capacities. 
Whether non-trivial competitive guarantees are possible for more general combinatorial structures (e.g., networks), even when the buyers have binary valuations and request bundles of equal size $d$, is posed as a challenging open question in \cite{chawla18}.

Hence, previous work motivates the following fundamental questions: 

\begin{quote}
\begin{enumerate}
\item \emph{What static pricing mechanisms best suit environments with complementarities?}

\item \emph{How fast does the competitive ratio improve with item multiplicity?}

\item \emph{Is there any standard framework for designing static bundle pricing mechanisms and establishing lower bounds in the presence of complementarities?}
\end{enumerate}
\end{quote}

\paragraph{Contribution in a Nutshell.}
In this work, we make substantial progress towards answering all the three questions above. On the positive side, we develop a unified framework for the design of bundle pricing mechanisms in the presence of complementarities. Our pricing mechanisms are \emph{static}, in the sense that the corresponding pricing menus are fully determined before any buyer's arrival, and \emph{anonymous}, in the sense that the pricing menu presented to
a buyer does not depend on their identity. We apply our framework  to the following domains:

\begin{enumerate}
\item $d$-single-minded combinatorial auctions, where requested bundles are of size at most $d$. 

\item General single-minded combinatorial auctions with arbitrary bundle sizes. 

\item Online network routing, where each buyer $b$ aims to route a unit of flow from a source node $s_b$ to a destination node $t_b$ in an edge-capacitated network. 
\end{enumerate}

We obtain competitive ratios that are $O(d)$ and $O(\sqrt{m})$, respectively, in the basic case of unit item capacities, and \emph{improve exponentially} with item capacities. Thus, one of our key conceptual findings is that \emph{bundle pricing becomes exponentially more powerful as item capacities increase} (resembling what happens for the standard offline optimization versions of single-minded combinatorial auctions and multicommodity network routing, see e.g., \cite{LehmannOS02,BartalGN03,KolliopoulosS04,LaviS11,BansalKNS12}). 
%

We complement our positive results with information-theoretic lower bounds for both the $d$-single-minded and the general single-minded settings that apply to any online algorithm. Our lower-bound construction is obtained by embedding the construction of \citet{babaioff} into the setting of single-minded combinatorial auctions. Our proof builds on an intriguing connection to a fundamental question in extremal combinatorics: ``\emph{what is the maximum number of pairwise qualitatively independent partitions of a ground set?}'', first asked by \citet{renyi}. 

\section{Model, Technical Overview and Main Results}
\label{sec:overview}

\paragraph{Model.}
We consider a set $\mathcal{M}$ of $m$ distinct items, each available in at least $B \in \mathbb{Z}_{>0}$ copies, and a set $\mathcal{N}$ of $n$ buyers, typically indexed by $b \in [n]$. For each buyer $b$, we are given a fixed set of items $S_b$ and a distribution $\mathcal{D}_b$ from which $b$'s value $v_b$ 
for their requested set $S_b$ is drawn. We assume that the distributions $\mathcal{D}_b$ are independent across buyers and are given to the mechanism in an explicit form, i.e., by the probabilities $q_{b,v}$ that buyer $b$ has value $v_b = v$ for getting their set $S_b$, for all buyers $b$ and all values $v$ in the support of $\mathcal{D}_b$ (which is assumed to be some fixed subset of non-negative integers). Each buyer $b$ has utility $v_b$ minus the price posted by the mechanism for their set $S_b$ (or a  superset $S \supseteq S_b)$, if $S_b$ (or some $S \supseteq S_b$) is allocated to $b$, and $0$ otherwise. 

We aim to maximize the expected social welfare, i.e., the expected total value of the buyers who receive their requested bundle, where expectation is taken over both the randomness of the mechanism and the buyers' realized values. The mechanism's competitive ratio is established against the expected social welfare of a ``prophet'', who collects the expected (over the buyers' realizations) total value of an optimal solution for the realized instance. Our mechanisms are online, in the sense that they have to irrevocably decide on the allocation of the requested bundles upon the agents' arrival, without any knowledge of the realized values of the agents arriving in the future. The order in which the online mechanism deals with the buyers is determined by an \emph{almighty adversary} that has access to the realized buyer values and the mechanism's internal randomness.  

\subsection{General Framework for Static Bundle Pricing}

Similar to previous work (see e.g., \cite{chawla18,multi-jiang,optimal-ocrs-singla}), we start with the ex-ante linear programming relaxation (a.k.a. EA-LP, formally stated in Section~\ref{sec:EA-LP}), which upper bounds the expected social welfare of the prophet. Before solving EA-LP, we scale down all item capacities by a factor~$\gamma$, creating a slack that later absorbs the demand fluctuations created by the randomness of the realization. The design of our pricing mechanism is guided by the optimal fractional solution to $\gamma$-scaled EA-LP. 

A notable feature of EA-LP's solution is its simple and rigid structure. For any buyer $b$, EA-LP allocates probability mass to values in an almost threshold-like form: below some cutoff value, the allocation is zero; above it, the allocation saturates the value distribution (resulting in so-called \emph{tight values}); and there is at most one \emph{crucial value} where the allocation is positive and non-tight (see Section~\ref{sec:structure}). Aggregating over all buyers interested in the same bundle~$S$, we define a canonical price level $w_S$ (a.k.a. the \emph{important value} of $S$) which dictates how the pricing menu treats~$S$.

This structure directly guides our menu construction (see  Section~\ref{sec:framework}). The mechanism posts a static, anonymous menu containing multiple ``copies'' of each bundle~$S$, priced either at the important value~$w_S$ or at the next higher value $w_S + 1$. We aim to mimic the optimal fractional solution of EA-LP as closely as possible: ideally, a buyer~$b$ with realized value~$v$ should receive an available copy of their desired bundle with probability $x_{b,v}/q_{b,v}$, independently of other buyers, where $x_{b,v}$ is the allocation mass assigned by EA-LP to the buyer-value pair $(b,v)$ and $q_{b,v}$ is the probability that $b$'s realized value is $v$. This would result in a competitive ratio of $O(\gamma)$, for a scaling factor $\gamma$ large enough to absorb demand fluctuations with constant probability.

A key challenge is that such independent acceptance probabilities cannot be implemented directly by a static and anonymous menu, in which the mechanism must commit ex-ante to offering a fixed number of copies of each bundle at a given price. To do so, guided by the solution of EA-LP, we post no copies of bundle~$S$ at prices below its important value $w_S$ and practically unlimited copies at the smallest price $w_S + 1$ that is strictly larger than $S$'s important value. Thus, we rely on the scarcity of high-value buyers (as quantified by the $\gamma$-scaled down capacity constraints of EA-LP) to limit the number of item copies allocated to the buyers. The number of copies offered at $S$'s important value $w_S$ requires more care. Offering too many copies risks exhausting item capacities and blocking high-value buyers later in the sequence, while offering too few copies may result in a significant loss relative to the fractional benchmark. To balance these effects, we offer $\lfloor x_{S,w_S} \rfloor$ additional copies of a bundle~$S$ at its important value~$w_S$, where $x_{S,w_S}$ denotes the total EA-LP allocation to buyers demanding bundle $S$ at value $w_S$. If $x_{S,w_S} < 1$, we randomize, with appropriate probability, between offering one additional copy of $S$ at price $w_S$ or not (see also Algorithm~\ref{alg:post-menu} in Section~\ref{sec:framework}). 

Our mechanism posts the static and anonymous menu described above, which consists of tuples $(S, w, c)$, denoting that $c$ copies of $S$ are available at price $w$. Let $w_b(S_b)$ be the lowest price at which a copy of $S_b$ is available in the mechanism's menu on the arrival of a buyer $b$. If $b$'s realized value $v_b$ is at least $w_b(S_b)$, a corresponding menu entry is allocated to $b$ and the number of $S_b$'s copies available at price $w_b(S_b)$ decreases by one. If all items $e \in S_b$ are still available (i.e., their residual capacity is positive), the bundle $S_b$ is allocated to $b$, the mechanism collects value $v_b$ and the residual capacity of all items $e \in S_b$ decreases by one. 

We highlight that the number of times an item $e \in \mathcal{M}$ is included in bundles offered by the mechanism's menu is typically much larger than $e$'s actual capacity. Hence, it may be the case that a buyer $b$ demands their bundle $S_b$ at a price currently available in the menu, but $b$ does not receive $S_b$ because some item $e \in S_b$ has been exhausted (i.e., $e$'s residual capacity is $0$). To deal with this situation in the analysis, we introduce a \emph{(capacity) unconstrained} version of our mechanism: a hypothetical process that completely ignores the item capacities and is only restricted by the mechanism's menu entries (guided as above by the fractional optimal solution to $\gamma$-scaled EA-LP). The unconstrained mechanism is not hard to analyze and satisfies two key properties:

\begin{enumerate}
\item Its expected social welfare is within a constant factor from the optimal value of the $\gamma$-scaled EA-LP, denoted $\fopt_\gamma$. In turn, $\fopt_\gamma$ is within $O(\gamma)$ of the prophet's expected social welfare. 

\item For every item~$e$, the number of $e$'s copies allocated to the buyers is at most $B/\gamma$ in expectation and (with an appropriate choice of $\gamma$) at most $B$ with constant probability. 
\end{enumerate}

The social welfare achieved by our mechanism equals the unconstrained welfare minus the expected value of \emph{blocked} buyers, namely buyers receiving a menu entry in the unconstrained mechanism, but not their desired bundle due to the item capacity constraints. Due to Property~(1), it suffices to upper bound the expected welfare lost due to blocked buyers by a constant fraction of $\fopt_\gamma$. To this end, we crucially exploit Property~(2). 

A key technical challenge in the analysis is that (due to the power of the almighty adversary) the availability of a buyer's desired bundle depends, in a quite delicate way, on both the buyers' realized values and the mechanism's randomization on the number of a bundle's $S$ copies posted at its important value $w_{S}$. Hence, we cannot deal with these two sources of randomness as independent anymore. Instead, we need to take some careful steps in order to separate the two sources of randomness. We believe that the careful analysis of our randomized menu construction might be of independent interest.

\paragraph{$d$-Single-Minded Setting.}
In the $d$-single-minded setting, we assume that $|S_b| \leq d$ for all buyers. 
Scaling the item capacities in EA-LP by a factor $\gamma = e(10d)^{1/B}$ and using a concentration argument, we show that any bundle allocated by the unconstrained mechanism is blocked due to item capacities with probability at most $0.1$. Employing a careful analysis of the event that a buyer $b$ with realized value equal to the important value of their desired bundle $S_b$ is blocked, we show (in Lemma~\ref{lem:d-blocked-small}) that the expected value lost due to blocked buyers is at most a $\fopt_\gamma / 10$ and obtain that: 

\begin{theorem}\label{thm:dsingle}
In the $d$-single-minded case with minimum item multiplicity $B$, our online static and anonymous bundle-pricing mechanism achieves a competitive ratio of $O(d^{1/B})$. 
\end{theorem}

\paragraph{General Single-Minded Setting.}
We next consider general single-minded buyers that may demand bundles of arbitrary size. Using more aggressive capacity scaling, with $\gamma = e(20m)^{1/(B+1)}$, in EA-LP and concentration, we show that the expected number of blocked buyers is at most $1/20$. 

A complication arises from the possibility of unbounded valuations: even a single blocked buyer with a very large value could dominate the welfare loss, rendering bounds on the number of blocked buyers insufficient. To address this issue, we augment the main mechanism with a simple \emph{small market} component that, with a carefully chosen probability, posts a single bundle at price $2\fopt_\gamma$. This auxiliary part captures a constant fraction of the welfare contributed by very large values, ensuring that heavy-tailed realizations do not overwhelm the performance guarantee. The main mechanism deals with the structured, high-probability portion of the market, while the small-market component isolates and stabilizes the contribution of extreme values. Thus, we obtain: 

\begin{theorem}\label{thm:general-single-minded}
For general single-minded combinatorial auctions with minimum item multiplicity $B$, our online static and anonymous bundle-pricing mechanism is $O(m^{1/(B+1)})$-competitive.
\end{theorem}

\paragraph{Online Network Routing.}
We extend our framework to online routing in a fixed directed network $G$. Instead of a bundle $S_b$, each buyer $b$ is now associated with two fixed known nodes, their source $s_b$ and their target $t_b$, and has value $v_b \sim \mathcal{D}_b$ for routing a unit of flow along any $s_b - t_b$ path. 

For the menu construction, we conceptually treat each source-target pair $(s,t)$ as an abstract bundle and group buyers according to their requested pair. Applying the same menu construction as before, we determine how many copies are posted at each price for each $(s,t)$ pair. Thus, we directly inherit Property~(1) of the unconstrained algorithm. 

An additional challenge is that we should select a specific $s - t$ path for each allocated $(s, t)$-bundle. To this end, we independently sample a path for each $(s, t)$-bundle available according to a carefully chosen distribution, which is derived from the solution to EA-LP. Using a roulette-style sampling argument, we show that our randomized routing scheme over selected $s-t$ paths  preserves the expected load guarantees on every edge and ensures Property~(2) of the unconstrained algorithm. Relying on the approach towards proving Theorem~\ref{thm:general-single-minded}, we obtain: 

\begin{theorem}\label{thm:graph}
For online routing in a network with $m$ edges, each with capacity at least $B$, our online static and anonymous bundle-pricing mechanism is $O(m^{1/(B+1)})$-competitive.
\end{theorem}

\paragraph{Incentive Compatibility.}
An additional desirable property of our menu-based pricing scheme is \emph{price subadditivity}. Namely, for every bundle~$S$, the posted price of~$S$ is no greater than the total price of any collection of bundles whose union includes~$S$. This guarantees that a buyer who desires~$S$ always prefers purchasing a copy of~$S$ rather than assembling a larger collection of bundles covering $S$. Enforcing price subadditivity, however, is nontrivial. The prices suggested by EA-LP need not satisfy this property, and the supports of the buyer value distributions may not contain appropriate price levels that would allow for a subadditive menu without sacrificing performance. To address this issue, we first augment the support of each buyer's value distribution with carefully chosen auxiliary price points. Then, before constructing our pricing menu, we apply a post-processing step to EA-LP solution that enforces subadditivity (the details can be found in Appendix~\ref{app:subadditivity}).

\subsection{Lower Bounds}

The proofs of our impossibility results follow a standard lower-bound template for prophet inequalities (see e.g., \cite{chawla2009,babaioff}). The buyers are partitioned into \emph{groups} of equal size~$t$. Buyers arrive sequentially, each has value $1$ with probability $1/t$ and $0$ otherwise. Their bundles are engineered so that we can select buyers from at most one group. Since the expected number of buyers realized in a fixed group is $1$, after committing to a group, the online algorithm can extract an expected value of $O(1)$ from that group. In contrast, if we have about $t^t$ groups, with constant probability, there exists a group where all $t$ buyers have value $1$ (a.k.a. a \emph{fully active} group), which results in an expected value of $\Theta(t)$ for the prophet. The main challenge is to \emph{embed} the aforementioned bundle structure in the setting of single-minded combinatorial auctions, where feasibility is induced by item capacities.

For unit capacities, we associate each group with a partition of the $m$ items into $t$ disjoint classes. The $t$ buyers in the group are single-minded, with buyer~$i$ demanding the $i$-th class. Disjointness within a partition guarantees that \emph{within a group}, all $t$ buyers can simultaneously be accepted. To prevent accepting buyers from \emph{different groups}, we choose a family of partitions that are \emph{pairwise qualitatively independent (QI)}, a notion introduced by \citet{renyi}, where every class of one partition intersects every class of every other partition. Thus, no feasible allocation can accept buyers from different groups, thus realizing the ``\emph{at most one group}'' constraint of the template. 

With $B$ copies per item, we have the constraint of ``\emph{at most $B$ groups}''. The corresponding combinatorial notion is that of \emph{$(B\!+\!1)$-way qualitative independence}: for any $B\!+\!1$ partitions and any choice of one class from each, their intersection is nonempty. Using such a family, any allocation that serves buyers from $B\!+\!1$ distinct groups must violate the capacity constraint of some item. Meanwhile, any fixed set of $B$ groups is simultaneously feasible, because each item appears in at most one demanded class per group.

\citet{poljak} gave a lower bound on the number of $r$-way qualitatively independent partitions of $m$ items into $t$ classes of 
\(
g(m, t, r) \geq \frac{r}{et}\cdot \exp\left({\frac{m}{rt^r}}\right)\,.
\)
Choosing $t = \floor{\lp(\frac{m}{2B\ln m}\rp)^{1/(B+2)}}$, we prove that there exists a $(B\!+\!1)$-way QI family of size at least $Bt^t$; thus, we can instantiate $Bt^t$ groups of $t$ buyers each. In our construction, each group is fully active with probability $\prn{1/t}^t$. Hence, the expected number of fully active groups is $B$, which implies that the expected social welfare of the prophet is $\Theta(tB)$. Due to our construction, any online algorithm is restricted to accepting buyers from at most $B$ groups and due to its online nature, can extract an expected value of $O(1)$ from each group. Therefore, 

\begin{theorem}\label{thm:lb-gen}
In the general single-minded setting with $m$ items, each with capacity at least $B$, the competitive ratio of any online algorithm is
$\Omega((m/\ln m)^{1/(B+2)})$.
\end{theorem}

For unit capacities, we can slightly improve the lower bound of Theorem~\ref{thm:lb-gen} first to $\Omega(m^{1/3})$, using Poljak’s explicit construction, and then to $m^{1/3+\Omega(1/\log\log m)}$, by relating pairwise qualitative independence to the Alon-Alweiss measure $\textsc{AAM}$ introduced in \cite{AA20, SVW22}, which yields larger QI families on slightly fewer items.

To obtain a lower bound for the $d$-single-minded setting, we introduce the notion of \emph{balanced QI partitions}, where every class has size at most $d$. We show that the probabilistic construction of \citet{poljak} extends to the balanced setting, yielding large balanced $(B\!+\!1)$-way QI families with all classes of size about $d = m/t$. Using 
\( t = \floor{(\frac{d}{2B\ln d})^{1/(B+1)}} \), 
we obtain that: 

\begin{theorem} \label{thm:lb-d}
In the $d$-single-minded setting with $m$ items, each with capacity at least $B$, the competitive ratio of any online algorithm is $\Omega((d/\ln d)^{1/(B+1)})$.
\end{theorem}

\section{Related Work}

\paragraph{Online Contention Resolution Schemes.} 
Our general approach is inspired by the framework of \emph{contention resolution schemes (CRSs)}, introduced by \citet{crs-chekuri}. A CRS is given a fractional point $x^*$ in the relaxation polytope and samples a random set $R$ according to the product distribution with marginals $x^*$. Since the resulting set $R$ may violate feasibility constraints, the CRS subsequently discards selected elements to obtain a feasible $R' \subseteq R$. Subsequently, \citet{ocrs} introduced the online counterpart of CRSs, termed \emph{online contention resolution schemes (OCRSs)}, which apply to a wide range of online settings, including prophet inequalities. \citet{ocrs} showed that the existence of a $c$-balanced OCRS implies a $c$-competitive prophet inequality. More recently, \citet{optimal-ocrs-singla} established the converse, i.e., that a $c$-competitive prophet inequality yields a $c$-competitive ex-ante OCRS. 

A central concept in the OCRS framework is \emph{selectability}. Informally, an OCRS is $c$-selectable with respect to a relaxation polytope~$P$ if for every point $x \in P$, every active element is included in the final solution with probability at least~$c$. However, restricting our attention to static bundle pricing mechanisms, selectability does not align well with anonymity. To illustrate, consider a single-item setting with two buyers, $a$ and $b$, each having deterministic value~$1$. The relaxation polytope admits a feasible fractional solution that assigns mass~$1/2$ to each buyer. Any static and anonymous pricing mechanism assigns a single price to the item, possibly drawn from a probability distribution. This price is accepted or rejected by both buyers. If the price is accepted, the adversary can always select buyer $a$ to receive the item, rendering meaningful selectability guarantees impossible.


\paragraph{Static versus Dynamic Pricing.}
A related line of work is due to \citet{rubinstein-beyond-matroids}, who studied prophet inequalities under general downward-closed feasibility constraints. He obtained an $O(\log n \cdot \log r)$-competitive algorithm, where $n$ denotes the number of arriving elements (corresponding to buyers in our setting) and $r$ the maximum cardinality of any feasible solution. Rubinstein's algorithm observes each realized value upon arrival and accepts it, if doing so does not decrease a carefully designed dynamic potential function. Since this potential function depends on the currently accepted set, the resulting acceptance rule is inherently dynamic. Recently, \citet{krysta} proved that any prophet inequality can be implemented by a (sequential, non-static and non-anonymous) posted-price mechanism. 

\paragraph{Lower Bounds.}
The general $\Omega(\log n/\log\log n)$ lower bound in prophet inequalities with $n$ buyers, which is used as a template for the proofs of theorems~\ref{thm:lb-gen}~and~\ref{thm:lb-d}, was introduced in other forms by \citet{babaioff} and \citet{chawla2009}. Lower bounds for prophet inequalities under the intersection of $d$ partition matroids also imply lower bounds in our online $d$-single-minded setting. To this end, \citet{klein-wein12} gave an $\Omega(\sqrt{d})$ lower bound and \citet{SVW22} strengthened this to $d^{1/2+\Omega(1/\log\log q)}$ using techniques from \cite{AA20}. 

\paragraph{Computational (In)approximability.}
Our mechanisms can be implemented in time polynomial in the number of buyers $n$ and the number of items $m$, as long as the maximum possible valuation is polynomially large in $n$ and $m$. Our approximation guarantees match (i) the computational complexity theoretic lower bounds of $\Omega(m^{1/(B+1)})$ for general single-minded combinatorial auctions \cite{BartalGN03} and $\Omega(\sqrt{m})$ for network routing in directed networks with unit edge capacities \cite{GuruswamiKRSY03}; and (ii) the best known approximation guarantees of $O(d^{1/B})$ for $d$-single-minded combinatorial auctions \cite{BansalKNS12} and $O(m^{1/(B+1)})$ for network routing in directed networks with edge capacity at least $B$ \cite{KolliopoulosS04}.

\section{Notation and Preliminaries}
\label{sec:prelim}

Our model is formally defined in Section~\ref{sec:overview}. We next introduce some additional notation. 

We let $[k] \coloneqq \{1, \ldots, k\}$, for any $k \in \mathbb{Z}_{>0}$.  We let $c(e) \in \mathbb{Z}_{>0}$ denote the multiplicity (a.k.a. capacity) of every item $e \in \mathcal{M}$. We assume that the buyers' value distributions $\mathcal{D}_b$ are discrete with common finite support $\Supp = \{0, \ldots, R \}$ with $R\in \mathbb{Z}_{>0}$. Thus, for all $b \in \mathcal{N}$ and $v\in \Supp$, $q_{b,v} > 0$. In Appendix~\ref{app:support-assumption}, we remove this assumption and allow buyers to have different supports. 

We let $\mathcal{D}\coloneq \prod_{b\in [n]} \mathcal{D}_b = \mathcal{D}_1\times \ldots \times \mathcal{D}_n$ and let $I\sim\mathcal{D}$ be an instance of buyers' values. 
For each buyer $b$, we denote by $v_b$ the random variable corresponding to its value.
For a set $S$, let $\mathcal{N}_S \coloneq \{b: S_b = S\}$ be the set of buyers that desire $S$. We let $q_{S,v} = \sum_{b\in \mathcal{N}_S} q_{b, v}$ be the total mass over buyers in $\mathcal{N}_S$ that have realized value $v$.
Whenever an index of the variables $x_{b, v}$ is omitted, we implicitly assume summation over that index (unless stated otherwise). E.g., $x_{b} = \sum_{v \in \Supp} x_{b, v}$ and $x_{v} = \sum_{b \in [n]} x_{b, v}$. 
For an event $\mathcal{E}$, we let $\ind{\mathcal{E}}$ denote its indicator function. We sometimes use $\ell_{A}(e)$ to denote the load of item $e$ due to buyers in set $A \subseteq \mathcal{N}$, i.e. $\ell_A(e) = \sum_{b\in A} \ind{e\in S_b}$.


\section{Ex-Ante Linear Program and Solution Structure}
\label{sec:EA-LP}

The competitive ratio of our mechanism is established against (and our pricing menu is guided by) an upper bound on the expected optimal social welfare $\opt$ given by the following ex-ante fractional relaxation of the social welfare maximization problem: 

\begin{align*}
    \max \quad &\sum_{b\in [n]}\sum_{v\in \Supp} x_{b, v} \cdot v \\ 
    \textnormal{s.t.} \quad &\sum_{b\in[n] : e\in S_b} \sum_{v\in \Supp} x_{b, v} \leq c(e) \qquad \forall e\in \mathcal{M} \\ 
    &0\leq x_{b, v} \leq q_{b,v} \qquad \forall b \in [n],\ \forall v\in \Supp
\end{align*}

In the above linear program, usually referred to as \emph{ex-ante LP} (or \emph{EA-LP}, in short), the capacity constraints are enforced only in expectation (ex-ante), so its solution represents a feasible ex-ante fractional allocation. We let $\fopt$ denote EA-LP's optimal value. Generally, ex-ante fractional relaxations upper bound the prophet's optimal expected social welfare $\opt$ (see e.g., \cite[Lemma~2.1]{chawla18}). For completeness, in Appendix \ref{app:proofs},  Lemma~\ref{lem:fopt-opt}, we prove that $\fopt \geq \opt$.

%

\subsection{Structure of Optimal Fractional Allocation}
\label{sec:structure}

Before formally describing our bundle pricing scheme, we provide intuition and establish some structural properties of the optimal EA-LP allocation, denoted $(x_{b,v})_{b\in [n], v\in \Supp}$ in what follows.  

If we did not require the mechanism to be anonymous or static, a natural approach would be to accept each new buyer $b$ with realized value $v$ with probability $x_{b,v}/q_{b,v}$. Namely, $x_{b,v}$ is interpreted as the ex-ante probability of allocating $S_b$ to buyer $b$, if $b$'s realized value is $v_b = v$. This allocation ensures that buyer $b$ is accepted with probability $x_b = \sum_{v \in \Supp} x_{b,v} = \sum_{v \in \Supp} q_{b,v} \cdot (x_{b,v}/q_{b,v})$ in total, matching the allocation suggested by EA-LP. However, this allocation is not anonymous (because buyers $b$ and $b'$ that demand the same bundle $S$ and have the same realized values $v$ may be accepted with different probability due to different priors $q_{b,v}$ and $q_{b',v}$). Moreover, such an allocation may not be feasible, because the item capacity constraints may be violated ex-post. In the following, we show how to implement a static and anonymous pricing scheme which mimics this allocation. 


\paragraph{Fractional Allocation per Buyer.}
Following up on the intuition above, we call a value $v\in \Supp$ \emph{tight} for a buyer $b$ if $x_{b,v} = q_{b,v} > 0$, and \emph{crucial}, if $0 < x_{b,v} < q_{b,v}$. Intuitively, we should always (resp. sometimes) accept $b$ if their value is tight (resp. crucial). 

For each buyer $b$, we consider the vector $X^{(b)} = ( x_{b,v} )_{v \in \Supp}$, whose entries correspond to the optimal fractional allocation to buyer $b$, ordered in increasing order of $v$. We have that:

\begin{lemma}\label{lem:structure-i}
For each buyer $b$, the vector $X^{(b)}$ consists of three consecutive blocks:
\begin{itemize}
    \item A (possibly empty) prefix of values $v$ where $x_{b,v}=0$, 
    \item followed by at most one value $v$ where $0 < x_{b,v} < q_{b,v}$ (the crucial value of $b$), 
    \item followed by a (possibly empty) suffix of values $v$ with $x_{b,v}=q_{b,v} > 0$ (the tight values of $b$).
\end{itemize}
\end{lemma}

\begin{proof}[Proof Sketch.]
The proof relies on a value exchange argument, showing that there is an optimal EA-LP solution where for any two buyers $b'$ and $b''$, with $S_{b''} \subseteq S_{b'}$, and for any two values $v'' \geq v'$, either $x_{b', v'} = 0$ or $x_{b'',v''} = q_{b'',v''}$ (or both -- the argument  also applies to the special case where $b'$ coincides with $b''$ and $v'' > v'$, which is relevant to this lemma). Otherwise, we could decrease $x_{b',v'}$ by $\delta = \min\{ x_{b', v'}, q_{b'',v''} - x_{b'',v''} \}$, while increasing $x_{b'',v''}$ by $\delta$. Thus, we increase the objective value of EA-LP, if $v'' > v'$ (or keep it unchanged, if $v'' = v'$), without violating any constraints. The item capacity constraints, in particular, remain satisfied, due to the assumption that $S_{b''} \subseteq S_{b'}$. The formal proof of Lemma~\ref{lem:structure-i} is based on this argument and can be found in Appendix~\ref{app:proofs}. 
\end{proof}

\paragraph{Fractional Allocation Aggregated by Bundle.}
To turn the intuition above into an anonymous pricing scheme, we aggregate the allocation of EA-LP by bundle and offer identical menu choices to all buyers demanding the same bundle. 
We recall that for any bundle $S$, $\mathcal{N}_S = \{b: S_b = S\}$ denotes the set of buyers interested in $S$. For each value $v \in \Supp$, we let $x_{S,v} = \sum_{b\in \mathcal{N}_S} x_{b, v}$ be the total fractional allocation of $S$ to buyers with realized value $v$. 

As before, we call a value $v\in \Supp$ \emph{tight} for a bundle $S$ if $x_{S,v} = q_{S,v} > 0$, and \emph{crucial}, if $0 < x_{S,v} < q_{S,v}$. For each bundle $S$, we consider the vector $X^{(S)} = (x_{S,v})_{v\in \Supp}$, whose entries correspond to the optimal fractional allocation of bundle $S$ at different values, ordered in increasing order of $v$. In Appendix~\ref{app:proofs}, we prove the following on the structure of $X^{(S)}$.

\begin{lemma}\label{lem:structure-S}
For each bundle $S$, the vector $X^{(S)}$ consists of three consecutive blocks:
\begin{itemize}
    \item a (possibly empty) prefix of values $v$ where $x_{S,v} = 0$,
    \item followed by at most one value $v$ with $0 < x_{S,v} < q_{S,v}$ (the crucial value of $S$),
    \item followed by a (possibly empty) suffix of values $v$ with $x_{S,v} = q_{S,v}$ (the tight values of $S$).
\end{itemize}
\end{lemma}

Our pricing menu is built around the \emph{important value} of a bundle $S$, denoted $w_S$, which is the largest value $w$ where $x_{S,w} < q_{S,w}$. If $x_{S,v} = q_{S,v}$ for all $v \in \Supp$, $S$'s important value is $w_S = 0$. We note that $w_S$ is the crucial value of $S$, if one exists, and $S$'s largest non-tight value (or $0$), otherwise.
\section{General Framework for Static and Anonymous Bundle Pricing}
\label{sec:framework}

We next present our general algorithmic framework, which serves as the basis for our static and anonymous bundle pricing mechanisms in the subsequent sections. At a high level, our framework proceeds in the following three steps.

\paragraph{Step 1: Scale Down EA-LP.}
\label{par:step1}
We first solve EA-LP, but with the item capacity constraints scaled down by a factor $\gamma > 1$ depending on the setting; i.e. $c(e)$ in the first EA-LP constraint becomes $c(e)/\gamma$. We let $\fopt_{\gamma}$ denote the new fractional optimum and observe that $\fopt_{\gamma} \geq \fopt/\gamma$.

\paragraph{Step 2: Construction of the Pricing Menu $M$.}
\label{par:step2}
We next create a (possibly randomized) menu $M$ from the solution of EA-LP, using Algorithm~\ref{alg:post-menu}. Our menu is randomized in the sense that the number of copies of a bundle $S$ at its crucial price may be a random variable (whose value is determined before any buyer arrives); importantly, the menu itself does not contain lotteries. The menu construction follows the structure of EA-LP's solution described by Lemma~\ref{lem:structure-S}. More specifically, we do not offer any copies of a bundle $S$ at values $v$ with $x_{S,v}=0$. We offer $\abs{\mathcal{N}_S}$ (i.e., practically unlimited) copies of $S$ at $S$'s smallest tight value $w_S+1$ in line~3 (so, availability of $S$ at price $w_S+1$ is only restricted by item capacities). Finally, if $S$'s important value $w_S$ is crucial, we offer a limited (and carefully chosen) number of additional copies of $S$ at price $w_S$. The number of such copies depends on $x_{S, w_S}$ and is the only case where we may use randomization. 

\begin{algorithm}[t]
\caption{Menu Construction \textbf{(Step 2)}}
\label{alg:post-menu}
\KwIn{Optimal solution $x$ of the $\gamma$-scaled-down EA-LP}

\For{each (desired) bundle $S$}{
    Let $w$ be the important value of $S$ (w.r.t. $x$)\;
    Offer $\abs{\mathcal{N}_S}$ (i.e., practically unlimited) copies of $S$ at price $w+1$\;

    \If{$S$ has a crucial value (which must be $w$) \textbf{and}
        $x_{S,w}\cdot w > \sum_{v \in \Supp: v > w} x_{S,v}\cdot v$}{
        \eIf{$x_{S,w} > 1$}{
            Offer $\lfloor x_{S,w} \rfloor$ additional copies of $S$ at price $w$\;
        }
        {
            With probability $\max\{x_{S,w},\, x_{S,w}/q_{S,w}\}$, offer one additional copy of $S$ at price $w$\;
        }
    }
}
\end{algorithm}

\paragraph{Step 3: Online Algorithm.}
We solve EA-LP and fix a pricing menu $M$, as described in steps~1~and~2 above. The values $v_b$ of all buyers $b \in \mathcal{N}$ are realized and an instance $I$ is fixed. Given $M$ and $I$, the almighty adversary determines the arrival order of the buyers in $\mathcal N$. Then, the online algorithm operates in a simple posted-pricing setting as follows: 
\begin{enumerate}
    \item Buyers arrive sequentially. When a buyer $b$ arrives, they examine the current menu $M$ and choose a subset $J$ (possibly empty) of available menu entries that maximizes their utility. 
    
    \item If $J\neq \emptyset$, for each chosen menu entry $(S_j, p_j, c_j) \in J$, we decrease its remaining multiplicity $c_j$ by $1$. If the residual capacities of all demanded items $e \in \bigcup_{j \in J} S_j$ are positive, we decrease all of them by $1$, buyer $b$ is deemed \emph{accepted} and the algorithm's social welfare increases by $v_b$. 
\end{enumerate}

We let $\alg_{M, I}$ denote either the set of buyers accepted by our algorithm or their total social welfare; the intended meaning will be clear from the context. Our goal is to prove that $\E{M, I}{\alg_{M, I}} \geq \alpha \cdot \fopt $ for some $\alpha>0$.

\paragraph{Price Subadditivity and Tie-Breaking.} We always assume that in step (3.1) above, every buyer $b$ selects a single menu entry that corresponds to their requested bundle $S_b$. In Appendix~\ref{app:subadditivity}, we show that this assumption is without loss of generality. More precisely, we show that our menu prices are subadditive; hence no buyer $b$ interested in $S_b$ has an incentive to select more menu entries whose union includes $S_b$. Moreover, we always assume that buyers break ties in favour of acceptance. Thus, if their value is $v$, they select a menu entry at price $v$. In Appendix~\ref{app:subadditivity}, we remove this assumption with a small price perturbation.

\subsection{The Unconstrained Algorithm and Its Properties}

We now introduce the \emph{Unconstrained Algorithm}, denoted $\ualg$, which completely ignores item capacities and is only restricted by the menu entries. Although the menu may induce total demand that exceeds some item capacities, $\ualg$ still accepts all such requests. The exact description of $\ualg$ can be derived from that of the online algorithm, in Step~3 above, by removing the phrase ``\emph{If the residual capacities\,$\ldots$\,of them by $1$}'' from step (3.2) (since $\ualg$ ignores item capacities). 

$\ualg$ is used only for the analysis. It may not be feasible, but it allows us to lower bound the expected social welfare of the online (capacity-feasible) algorithm. We use $\ualg_{M,I}$ to denote either the set of buyers accepted by the unconstrained algorithm or their total social welfare. $\ualg$ is just an analytical tool; it is supposed to operate on the same instance $I$, with the same buyer arrival order and with the same menu $M$ as the actual online algorithm $\alg$. Therefore, every buyer that belongs to $\alg_{M, I}$ must also belong to $\ualg_{M, I}$. 

We start with the first key property of $\ualg$, proven in Appendix~\ref{app:gen-frame}. 

\begin{lemma}[Property 1] \label{lem:ualg}
$ \E{M, I}{\ualg_{M, I}}\geq \fopt_\gamma/8$. 
\end{lemma}

Since $\ualg$ ignores item capacities, we need to bound how many copies of each item are allocated by $\ualg$. Let $L_{M, I}(e)$ be the random variable that denotes the number of copies of item $e$ allocated by $\ualg_{M, I}$. We define a collection of random variables that (as we will show) upper bound $L_{M, I}(e)$ and allow us to apply concentration inequalities.

\begin{definition} \label{def:Z}
For a bundle $S$, let $k_S$ be the (deterministic) number of additional copies of $S$ posted at $S$'s important value $w_S$ by Algorithm~\ref{alg:post-menu} ($k_S=\floor{x_{S,w_S}}$, if line~6 is executed, and $k_S=0$ otherwise), and let $Y_S\in\{0,1\}$ be the (possibly random) indicator that one additional copy of $S$ is posted at price $w_S$ (line~8). We define the following random variables for each bundle $S$ and each item $e$:
\[
Z_{M,I}(S) \coloneqq \sum_{b\in \mathcal{N}_S} \ind{v_b > w_S} + k_S + Y_S\cdot \ind{\exists b \in \mathcal{N}_S\,:\, v_b = w_S} 
\quad\text{and}\quad
Z_{M, I}(e) \coloneqq \sum_{S:e\in S} Z_{M,I}(S). \]
\end{definition} 

The following, proven in Appendix \ref{app:gen-frame}, shows that $\ualg$ respects item capacities in expectation. 

\begin{lemma}[Property 2] \label{lem:L(e)}
For every item $e \in \mathcal{M}$ and any scaling factor $\gamma>1$, it  holds that
\[
    L_{M,I}(e) \le Z_{M, I}(e)
    \quad\text{and}\quad
    \E{M, I}{L_{M,I}(e)}\le \E{M, I}{Z_{M, I}(e)} \le c(e) / \gamma\,.
\]

In addition, $Z_{M, I}(e)$ is a sum of negatively associated (NA) random variables.
\end{lemma}

Our analysis upper bounds the lost value from buyers that belong to $\ualg$ but not to $\alg$. Let 
\begin{equation}\label{eq:blocked}
\textsc{Blocked}_{M,I} = \left\{ b \in \mathcal{N} : b \in \ualg_{M,I} \text{ and } b 
\notin \alg_{M,I} \right\},
\end{equation}
be the set of not accepted buyers in $\ualg$ and let $\alg_{M, I}$ denote the set and the total value of the buyers that our algorithm accepts under realization $I$ and menu $M$. It holds that 
\begin{align}\label{eq:mainIdea}
    \E{M,I}{\alg_{M,I}}
    \;=\;
    \E{M,I}{\ualg_{M,I}}
    \;-\;
    \E{M,I}{\sum_{b \in \textsc{Blocked}_{M,I}} v_b }
    \tag{$\star$}
\end{align}
\section{\emph{d}-Single-Minded Combinatorial Auctions: The Proof of Theorem~\ref{thm:dsingle}}
\label{sec:dsingle}

Next, we apply our framework to the case where every buyer demands a single bundle $S_b \subseteq \mathcal M$ of size at most $d$ and prove Theorem~\ref{thm:dsingle}.
To this end, we follow the general framework presented in Section~\ref{sec:framework}. We begin by solving the $\gamma$-scaled down EA-LP, with $\gamma = e (10d)^{1/B}$, construct the randomized menu $M$ as in Algorithm~\ref{alg:post-menu} and apply the online algorithm presented in Step~3. 
 
\paragraph{Competitive Ratio.}
Lemma~\ref{lem:L(e)} and the choice of $\gamma = e(10d)^{1/B}$ immediately imply that:
\begin{corollary}\label{cor:d-single-minded}
For every item $e \in \mathcal M$, 
\(
    \E{M,I}{ Z_{M,I}(e) }
    \leq
    c(e)/\gamma
    =
    c(e) / e(10d)^{1/B}\,.
\)
\end{corollary}

We first show that the probability that a bundle $S$ allocated by the unconstrained mechanism is blocked due to item capacities is at most $0.1$.
Since $Z_{M,I}(e)$ is a sum of negatively associated binary random variables and upper bounds $L_{M, I}(e)$, as shown in Lemma \ref{lem:L(e)}, we apply the standard Chernoff bound to $Z_{M,I}(e)$.  Using $\mu_e = \E{M,I}{Z_{M,I}(e)} \leq c(e)/\gamma$, for every $e \in \mathcal M$, we obtain that
\[
    \Prob{ Z_{M,I}(e) \ge c(e) }
    \;\le\;
    \left( \frac{e \mu_e}{c(e)} \right)^{c(e)}
    \;\le\;
    \left( \frac{e}{\gamma} \right)^{c(e)}
    \;\le\; \frac{1}{10d} \,,
\]
by our choice of $\gamma$. Applying a union bound and $|S| \leq d$, we get that for every desired bundle $S$, 
\[
\Prob{\exists e\in S:\, Z_{M,I}(e)\ge c(e)}
\le \sum_{e\in S}\Prob{Z_{M,I}(e)\ge c(e)}
\le \abs{S}\cdot \frac{1}{10d}
\le 0.1,
\]

When a buyer arrives, they may prefer some menu entry (so they are accepted by $\ualg$), but they may not be accepted by the actual algorithm because one of the items in their bundle has exhausted its capacity. Then, we say that the buyer is \emph{blocked} by the algorithm (see also \eqref{eq:blocked}). The following bounds the expected social welfare lost due to blocked buyers and implies Theorem~\ref{thm:dsingle}.

\begin{lemma}\label{lem:d-blocked-small}
We have that 
\( \quad {\displaystyle
    \E{M,I}{\sum_{b \in \textsc{Blocked}_{M,I}} v_b}
    \;\le\;
    0.1\;
    \fopt_\gamma\,.}
\)
\end{lemma}

\begin{proof}
The event that a buyer $b$ is blocked depends, in a delicate way, on their value $v_b$ and the realized values of all other buyers (the realized values of all buyers comprise the instance $I$, both the buyers' arrival order determined by the almighty adversary and the algorithm's set of accepted buyers when $b$ arrives depend on $I$) and the randomized construction of the algorithm's menu $M$. The proof decouples the above sources of randomness via carefully selected events. 

We fix a bundle $S$ with important value $w = w_S$. We let $\mathcal{A}=\mathcal{N}_S$ be the set of buyers interested in $S$. We let $\mathcal{E}_b$ denote the event that buyer $b$ is not accepted due to exhausted item capacities. 

We upper bound the expected welfare lost due to blocked buyers in $\mathcal{A}$ by analyzing separately the loss due to blocked buyers with realized values larger than (resp. equal to) the important value $w$, denoted $T_{>w}$ and $T_{=w}$, respectively. Formally, we let 
\[
T_{>w}\coloneqq \E{M,I}{\sum_{b \in \mathcal{A}} v_b \ind{v_b > w}\ind{\mathcal{E}_b}}
\text{ and }T_{=w}\coloneqq w\cdot \E{M,I}{\sum_{b \in \mathcal{A}} \ind{v_b = w}\ind{b\in \ualg_{M,I}}\ind{\mathcal{E}_b}}
\]
It suffices to show that $T_{>w} + T_{=w} \leq 0.1\,\sum_{v > w} v\, x_{S,v} + 0.1\,w\,x_{S, w}$. Then, the lemma follows by applying linearity of expectation over all desired bundles $S$. 

For a fixed bundle $S$ and each item $e \in S$, we distinguish between the demand for $e$ due to allocated copies of bundles different from $S$, denoted $Z^{\neg S}_{M,I}(e) \coloneqq \sum_{T\neq S:\ e\in T} Z_{M,I}(T)$, and the demand for $e$ due to allocated copies of $S$, denoted $Z_{M,I}(S)$ (see also Definition~\ref{def:Z}). Then, $Z_{M,I}(e)=Z^{\neg S}_{M,I}(e)+Z_{M,I}(S)$. We observe that $Z^{\neg S}_{M,I}(e)$ depends only on the realized values of buyers not in $\mathcal{A}$ and the randomness in the menu construction for bundles $T\neq S$; it does not depend on $v_b$, of any $b \in \mathcal{A}$, or on $Y_S$. 

\paragraph{Bounding $T_{>w}$.}
To apply linearity of expectation in the definition of $T_{>w}$, we fix a buyer $b \in\mathcal{A}$. To decouple $b$'s realized value $v_b$ from the events $\ind{v_b > w}$ and $\ind{\mathcal{E}_b}$ (and these two from each other) in the resulting expectation, we define the following event as a substitute for $\mathcal{E}_b$:
\[
\mathcal{G}_b \;\coloneqq\;
\left\{\exists e\in S:\ 
Z^{\neg S}_{M,I}(e)
+\!\!\sum_{b'\in \mathcal{A}\setminus\{b\}}\!\!\!\ind{v_{b'}>w}
+ k_S
+ Y_S\cdot \ind{\exists b'\in \mathcal{A}\setminus\{b\}: v_{b'}=w} \,\ge\, c(e)\right\}
\]
(see Definition~\ref{def:Z} for the definition of $k_S$ and $Y_S$). Intuitively, the event $\mathcal{E}_b$ (i.e., that $b$ is rejected) implies $\mathcal{G}_b$. However, $\mathcal{G}_b$ does not depend on the buyers' arrival order or on $b$'s realized value $v_b$ (because it only depends on the realized values of buyers in $\mathcal{A}\setminus\{b\}$ and on the randomness in the menu construction). The following establishes the key property of (and the intuition behind) $\mathcal{G}_b$.

\begin{claim}
    $\ind{v_b>w}\ind{\mathcal{E}_b}\le \ind{v_b>w}\ind{\mathcal{G}_b}$.
\end{claim}
\begin{proof}
Whenever $v_b>w$, buyer $b$ is accepted by $\ualg_{M,I}$ because $M$ includes practically unlimited copies of $S$ at price $w+1$. We show that if additionally $\mathcal{E}_b$ occurs, then $\mathcal{G}_b$ occurs as well. If $\mathcal{E}_b$ occurs, when $b$ 
arrives there exists some $e\in S$ whose capacity has been exhausted by previously accepted buyers (other than $b$). The contribution of buyers not in $\mathcal{A}$ to the load of $e$ is at most
$Z^{\neg S}_{M,I}(e)$. The number of previously accepted buyers from $\mathcal{A}\setminus\{b\}$ is at most $\sum_{b'\in \mathcal{A}\setminus\{b\}}\ind{v_{b'}>w}$ plus at most $k_S$ buyers with value $w$ (getting the $k_S$ deterministic copies of $S$ with price $w$) plus possibly one additional buyer in $\mathcal{A}\setminus\{b\}$ with value $w$, 
if $Y_S=1$ and such a buyer exists. 
\end{proof}

Also, by Definition~\ref{def:Z}, for every $e\in S$ and every menu $M$ and instance $I$, we have that 
\[
Z^{\neg S}_{M,I}(e)
+\underbrace{\sum_{b'\in \mathcal{A}\setminus\{b\}}\!\!\!\ind{v_{b'}>w}
+ k_S
+ Y_S\cdot \ind{\exists b'\in \mathcal{A}\setminus\{b\}: v_{b'}=w}}_{\leq Z_{M,I}(S)}
\,\leq\, Z^{\neg S}_{M,I}(e)+Z_{M,I}(S) = Z_{M,I}(e)
\]
Therefore, 
\(
\Prob{\mathcal{G}_b} 
\le \Prob{\exists e\in S:\ Z_{M,I}(e)\ge c(e)} \le 0.1
\). 

Using independence of $\mathcal{G}_b$ and $v_b$, we obtain that
\[
\E{M,I}{v_b\ind{v_b>w}\ind{\mathcal{E}_b}}
\le
\E{M,I}{v_b\ind{v_b>w}\ind{\mathcal{G}_b}}
=
\Prob{\mathcal{G}_b}\cdot \sum_{v>w} v\,q_{b,v}
\le
0.1\sum_{v>w} v\,q_{b,v}.
\]

Applying linearity of expectation over all buyers $b\in\mathcal{A}$ and using that $w$ is the important value of $S$ (and thus, all values $v > w$ are tight for $S$), we conclude that 
\[
T_{>w}\ \le\ 0.1\sum_{b\in\mathcal{A}}\sum_{v>w} v\,q_{b,v}
\ =\ 0.1\sum_{v>w} v\,q_{S,v}
\ =\ 0.1\sum_{v>w} v\,x_{S,v},
\]

\paragraph{Bounding $T_{=w}$.}
We need to show that the expectation in the definition of $T_{=w}$ is at most $0.1\,x_{S,w}$. We let $p_S \coloneqq \Prob{Y_S=1}$ denote the probability that an additional copy of $S$ is posted at price $w$ (in line~8 of Algorithm~\ref{alg:post-menu},  see also Definition~\ref{def:Z} for $Y_S$). We observe that conditional on $k_S > 0$, $p_S=0$. Hence, we get that
%
$k_S + p_S\cdot \min\{1,q_{S,w}\} \le x_{S,w}$, which holds unconditionally. 

We let $N_w\coloneqq \sum_{b\in\mathcal{A}}\ind{v_b=w}$ be the random variable that denotes the number of buyers $b \in \mathcal{A}$ with $v_b = w$. A difficulty in bounding $T_{=w}$ arises from the need to consider all buyers $b \in \mathcal{A}$ with $v_b = w$ together, so that we can apply the upper bound of $k_S + Y_S \cdot \ind{N_w\ge 1}$ on the number of such buyers accepted by $\ualg$ (see also Definition~\ref{def:Z} for $k_S$). Hence, we define the following event as a substitute of the event that some buyers in $b \in \mathcal{A}$ with $v_b = w$ are not accepted: 
\[
\mathcal{F}_S \coloneqq
\left\{\exists e\in S:\ Z^{\neg S}_{M,I}(e) + \sum_{b\in \mathcal{A}}\ind{v_b>w} + k_S \ \ge\ c(e)\right\}
\]
We observe that $\mathcal{F}_S$ does not depend on the buyers' arrival order and that its occurrence is negatively associated with $N_w$ (because each buyer $b \in \mathcal{A}$ has either $v_b \le w$ or $v_b > w$, but not both). 
Whenever there exists a buyer $b \in \mathcal{A}$ with $v_b = w$ not accepted by the algorithm (i.e., the corresponding event $\mathcal{E}_b$ occurs), the event $\mathcal{F}_S$ occurs as well. 
Moreover, for every item $e\in S$, every realized instance $I$ and every realized menu $M$, the sum accounting for the load of $e$ in the definition of $\mathcal{F}_S$ does not exceed $Z_{M,I}(e)$. Therefore, \(
\Prob{\mathcal{F}_S} \le \Prob{\exists e\in S:\ Z_{M,I}(e)\ge c(e)} \le 0.1 \).

The total number of buyers $b\in \mathcal{A}$ with $v_b = w$ accepted by $\ualg$ (such buyers $b$ may trigger an event $\mathcal{E}_b$ in the definition of $T_{=w}$) is at most $k_S$ (via the $k_S$ deterministic copies of $S$ posted at price $w$ in  Algorithm~\ref{alg:post-menu}, line~6) plus at most one additional such buyer, in case where $Y_S=1$ and $N_w\ge 1$. Hence, for every realized instance $I$ and every realized menu $M$, we have that: 
\[
\sum_{b \in \mathcal{A}} \ind{v_b=w}\ind{b\in \ualg_{M,I}}\ind{\mathcal{E}_b}
\ \le\
\Big(k_S + Y_S\cdot\ind{N_w\ge 1}\Big)\cdot \ind{\mathcal{F}_S}.
\]

Taking expectations and using that the realization of $Y_S$ depends only on the randomness of an independent coin flip (for bundle $S$) in Algorithm~\ref{alg:post-menu}, line~8 (so $Y_S$ does not depend on the realized instance $I$ or on $\mathcal{F}_S$) and that $\ind{\mathcal{F}_S}$ and $\ind{N_w\ge 1}$ are negatively associated,) we obtain that:
\begin{equation}\label{eq:exp-tw}
\E{M,I}{\sum_{b \in \mathcal{A}} \ind{v_b=w}\ind{b\in \ualg_{M,I}}\ind{\mathcal{E}_b}}
\ \le\
\Big(k_S + p_S\cdot \E{M,I}{\ind{N_w\ge 1}}\Big) \cdot \Prob{\mathcal{F}_S}
\end{equation}

We observe that $\E{M,I}{\ind{N_w\ge 1}} \leq \min\{ 1, q_{S, w} \}$, since $\ind{N_w\ge 1}\le 1$ and $\ind{N_w\ge 1}\le N_w$, with  $\E{M,I}{N_w} = \sum_{b\in\mathcal{A}}\Prob{v_b=w} = q_{S, w}$. Therefore, the right-hand-side of \eqref{eq:exp-tw} can be upper bounded by $(k_S + p_S\cdot \min\{1, q_{S,w}\})\cdot \Prob{\mathcal{F}_S}$.

 Using that $\Prob{\mathcal{F}_S} \leq 0.1$ and that $k_S + p_S\cdot \min\{1,q_{S,w}\} \le x_{S,w}$, we conclude $T_{=w} \leq 0.1\,w\,x_{S, w}$\,. 
\end{proof}

Applying Lemma~\ref{lem:d-blocked-small}, \eqref{eq:mainIdea} and Lemma~\ref{lem:ualg}, we conclude the proof of Theorem~\ref{thm:dsingle}:
\[
    \E{M,I}{\alg_{M,I}}
    \;\geq\;
    \E{M,I}{\ualg_{M,I}}
    \;-\;
    0.1\,\fopt_{\gamma} \geq \frac{\fopt_\gamma}{40}
    \;\geq\; \frac{\fopt}{40e (10d)^{1/B}}
\]

\section{General Single-Minded Combinatorial Auctions: The Proof of Theorem~\ref{thm:general-single-minded}}
\label{sec:general-single-minded}

Next, we apply our general algorithmic framework to the case where the buyers are single-minded over arbitrary bundles and prove Theorem~\ref{thm:general-single-minded}.
%

The main challenge in the proof of Theorem~\ref{thm:general-single-minded} is that in certain (very low probability) instances, the realized value of some buyer may be extremely large, much larger than our fractional benchmark $\fopt_\gamma$. Then, the almighty adversary can arrange that the buyers with extremely large realized value arrive very last, so that it maximizes the probability that they are rejected by the algorithm (while they are accepted by $\ualg$). To deal with this case, we introduce an auxiliary ``\emph{small market}'' mechanism $\alg_2$, which posts a single bundle consisting of entire $\mathcal{M}$ at price $2\,\fopt_\gamma$. Our final mechanism randomizes between the auxiliary mechanism $\alg_2$ (with probability $2/3$) and our standard ``\emph{large market}'' mechanism $\alg_1$ obtained from the general framework in  Section~\ref{sec:framework} (with probability $1/3$). For $\alg_1$, we solve the $\gamma$-scaled down EA-LP, with $\gamma = e(20m)^{1/(B+1)}$, construct the randomized menu $M$ as in Algorithm~\ref{alg:post-menu} and apply the online algorithm presented in Step~3.


As in Section~\ref{sec:dsingle}, we need to upper bound the expected social welfare lost due to blocked buyers. To this end, we account for the contribution of buyers with extremely large and moderate realized value separately. Formally, we decompose (and upper bound) $\E{M,I}{\sum_{b \in \textsc{Blocked}_{M,I}} v_b}$ as follows: 
\begin{multline}\label{eq:general}
\E{M,I}{\sum_{b \in \textsc{Blocked}_{M,I}}
        v_b\,\ind{v_b < 2\fopt_\gamma}}
+
\E{M,I}{\sum_{b \in \textsc{Blocked}_{M,I}}
        v_b\,\ind{v_b \ge 2\fopt_\gamma}}
\\ 
\leq
2\,\fopt_\gamma \cdot \E{M,I}{\abs{\textsc{Blocked}_{M,I}}}
+
\E{M,I}{\sum_{b \in \mathcal{N}}
        v_b\,\ind{v_b\ge 2\fopt_\gamma}}
\end{multline}

\paragraph{Bounding the Number of Blocked Buyers.} 
We begin by analyzing the first term. The following, whose detailed proof is deferred to Appendix~\ref{app:proofs}, bounds the expected number of blocked buyers.

\begin{lemma}\label{lem:blocked-small}  
We have that 
\( \E{M,I}{\abs{\textsc{Blocked}_{M,I}}}
    \;\le\; 1/20  \)\,.
\end{lemma}

\begin{proof}[Proof Sketch]
A buyer is blocked only if some item is oversold, hence
\[
|\textsc{Blocked}_{M,I}|\le \sum_{e\in\mathcal M}(L_{M,I}(e)-c(e))_+
= \sum_{e\in\mathcal M}\sum_{j\ge1}\ind{L_{M,I}(e)\ge c(e)+j}.
\]
Taking expectations and using Lemma \ref{lem:L(e)}, a Chernoff bound for negatively associated (NA) variables with $\mu_e = \E{M,I}{Z_{M,I}(e)} \leq c(e)/\gamma$, gives that for every item $e \in \mathcal M$, 
\[
\Prob{L_{M,I}(e)\ge c(e)+j}\leq \Prob{Z_{M,I}(e)\ge c(e)+j}\le \frac{1}{20m}e^{-j}\,.\]
Summing over all $j\ge1$ yields $\E{M, I}{(L_{M,I}(e)-c(e))_+}\le 1/(20m)$. Summing over all items $e \in \mathcal{M}$ concludes the proof of the lemma.
\end{proof}

Plugging Lemma~\ref{lem:blocked-small} into \eqref{eq:general} and using \eqref{eq:mainIdea} and Lemma~\ref{lem:ualg}, we obtain that:
\begin{equation}\label{eq:large_market}
    \E{M,I}{\alg_{M,I}}
    \;\ge\;
        \frac{\fopt_\gamma}{40}
      - \sum_{b \in \mathcal{N}} \E{M,I}{v_b\,\ind{v_b\ge 2\fopt_\gamma}}\,.
\end{equation}

\paragraph{Dealing with Extremely Large Values.} 
The negative term of \eqref{eq:large_market} corresponds to buyers with extremely large realized values. To deal with them, we introduce auxiliary mechanism $\alg_2$, which posts a single menu entry consisting of all items,  priced at $2\,\fopt_\gamma$. We distinguish two cases.

\textbf{Case 1:} $\Prob{\exists b : v_b \ge 2\,\fopt_\gamma} > 1/2$.
Then, simply posting one copy of entire $\mathcal{M}$ at price $2\,\fopt_\gamma$ yields an expected social welfare larger than $\fopt_\gamma$, ensuring an $O(\gamma)$-approximation directly.

\textbf{Case 2:} $\Prob{\exists b : v_b \ge 2\fopt_\gamma} \leq 1/2$.
The following shows that in this case, the expected social welfare of $\alg_2$ is at least half of the negative term of \eqref{eq:large_market}. 

\begin{lemma}\label{lem:alg2-captures-large}
If \,$\Prob{\exists b : v_b \ge 2\,\fopt_\gamma} \leq 1/2$, then\ \,$
        \E{M,I}{\alg_2}
        \ge
        \frac{1}{2}\,\sum_{b \in \mathcal{N}}
        \E{M,I}{ v_b\,\ind{v_b\ge 2\,\fopt_\gamma} }.
    $
\end{lemma}

\begin{proof}
For every fixed realized instance $I$, we have that: 
$$
    \alg_2
    =
    \sum_{b \in \mathcal{N}}
        v_b\,\ind{v_b\ge 2\,\fopt_\gamma} \,
        \ind{\text{the posted copy is still available for } b}\,.
$$
Whenever all other buyers $b'\neq b$ have $v_{b'}<2\fopt_\gamma$, the posted copy of $\mathcal{M}$ is available for $b$. Thus,
\[
    \alg_2
    \;\ge\;
    \sum_{b \in \mathcal{N}} v_b\,\ind{v_b\ge 2\,\fopt_\gamma}\,
        \ind{\forall b'\neq b : v_{b'}<2\,\fopt_\gamma}\,.
\]
Taking expectations and using independence of the realized values of different buyers, 
\begin{align*}
    \E{I}{\alg_2}
    &\;\ge\;
    \sum_{b}
        \E{I}{v_b\ind{v_b\ge 2\,\fopt_\gamma}}\,
        \Prob{\forall b' \neq b : v_{b'}<2\,\fopt_\gamma}\\
    &\;\ge\;
    \frac{1}{2} \sum_b \E{I}{v_b\ind{v_b\ge 2\,\fopt_\gamma}},
\end{align*}
where the last inequality holds due to the definition of Case 2.
\end{proof}

\paragraph{Combined Algorithm and Competitive Ratio.} 
Let $\alg_1$ denote the mechanism generated by the EA-LP-based menu, and let the final algorithm $\alg$ run $\alg_1$ with probability $1/3$ and $\alg_2$ with probability $2/3$. Combining the previous bounds, we conclude the proof of Theorem~\ref{thm:general-single-minded}:
\begin{align*}
    \E{}{\alg}
    &\ge
    \frac{1}{3}
    \left(
        \frac{\fopt_\gamma}{40}
        -
        \sum_b \E{}{v_b\ind{v_b\ge 2\,\fopt_\gamma}}
    \right) + 
    \frac{2}{3}\cdot \frac{1}{2}
    \sum_b \E{}{v_b\ind{v_b
    \ge 2\,\fopt_\gamma}}
    \\
    &= \frac{\fopt_\gamma}{120}
    \;\geq\;
    \frac{\fopt}{120e(20m)^{1/(B+1)}}.
\end{align*}

\section{Online Network Routing: The Proof of Theorem~\ref{thm:graph}}
\label{sec:graph}

We next extend our framework to the slightly more general setting where items correspond to edges and buyers correspond to routing requests in a network. Let $G = (V,E)$ be directed network equipped with an integer capacity function $c : E \to \mathbb{Z}_{>0}$. Each buyer $b\in[n]$ is associated with a publicly known source-target pair $(s_b,t_b)$. Upon arrival, buyer $b$ draws a value $v_b$ from a known distribution $\mathcal{D}_b$. Buyer $b$ derives value $v_b$ if they receive bundle of edges that contains a path connecting $s_b$ to $t_b$; otherwise their value is zero. 

Let $\mathcal{P}_{s,t}$ be the set of all simple $s$-$t$ paths in $G$ and let $\mathcal{P}_{b} = \mathcal{P}_{s_b, t_b}$. As above, for each $b$, we let $q_{b,v} = \Prob{v_b = v}$. We write $p\ni e$ to denote that a path $p$ uses edge $e$.

As in the single-minded setting, we compare against the optimal prophet via the following fractional relaxation. The EA-LP now contains one variable for every possible routing choice:
\[
\begin{aligned}
\max\quad & \sum_{b\in[n]} \sum_{v\in\Supp} \sum_{p\in\mathcal{P}_b}
    x_{b,v,p}\cdot v \\
\text{s.t.}\quad 
& \sum_{b\in[n]} \sum_{v\in\Supp} 
    \sum_{p\in\mathcal{P}_b: e\in p} x_{b,v,p} 
    \;\le\; c(e), && \forall e\in E, \\
& \sum_{p\in\mathcal{P}_b} x_{b,v,p} \;\le\; q_{b,v}, 
    && \forall b\in[n],\, v\in\Supp, \\ 
& x_{b,v,p} \ge 0, && \forall b\in[n],\, v\in\Supp,\, p\in\mathcal P_b.
\end{aligned}
\]
As before, we solve (and use) a $\gamma$-scaled down of EA-LP, with $c(e)/\gamma$, instead of $c(e)$, on the rhs of the first constraint. We let $\fopt$ (resp. $\fopt_\gamma$) denote the value of the optimal fractional (resp. scaled down) EA-LP solution.

\paragraph{Buyer Types and Generalized Bundles.} 
A key observation is that any $s$-$t$ path satisfies all buyers with the same source-target pair $(s,t)$. This motivates treating each pair $(s,t)$ as a \emph{type} whose feasible bundles are all $s$-$t$ paths in $\mathcal{P}_{s,t}$. We introduce the notion of a \emph{generalized bundle} $B_{s,t}$, representing the act of serving an $(s,t)$-buyer. Each copy of $B_{s,t}$ is instantiated as a specific path $p \in \mathcal{P}_{s,t}$; different copies of bundle $B_{s,t}$ may correspond to different $s-t$ paths. The path associated with each copy is sampled before any buyer arrives from a distribution over $\mathcal{P}_{s,t}$ defined below.

For simplicity, we assume that every buyer $b$ with $(s_b,t_b) = (s,t)$ can only purchase copies of $B_{s,t}$%
\footnote{In Appendix~\ref{app:graph-subadd}, we remove this assumption by proving that the menu prices are subadditive; so the buyers never have an incentive to purchase a collection of paths whose union includes an $s-t$ path, instead of purchasing a copy of $B_{s,t}$.}.
Then, buyers of the same type behave like single-minded buyers w.r.t. a generalized ``bundle'' $B_{s,t}$. We may therefore group them as a single demand type and randomize over how many copies of $B_{s,t}$ to create, just as we randomize over bundle copies in the single-minded setting.

Crucially, all arguments that relate $\fopt$ to the expected value of the unconstrained algorithm (Property~1, Lemma~\ref{lem:ualg}) depend only on the total fractional allocation to each generalized ``bundle'' $B_{s,t}$ and not on the bundle's internal structure (i.e., how different copies of $B_{s,t}$ are routed through different $s-t$ paths). After grouping buyers by type, the analysis of $\ualg$ from the single-minded case carries over without any change. It remains to specify how (i)~we construct the menu for each $(s,t)$ type; and (ii) we route buyers who purchase $B_{s,t}$ so that the load bounds (Property~2, Lemma~\ref{lem:L(e)}) are preserved at the edge level.

\paragraph{Menu Construction and Routing.} 
We now describe the construction of the menu $M$ in the online network routing setting. For each pair $(s,t)$ with at least one buyer, we let
\(
\mathcal{N}_{s,t}
    \coloneqq \{ b \in [n] : (s_b,t_b) = (s,t)\}
\)
be the set of buyers of type $(s,t)$. We define
\(
    X_{(s,t), v}
    =
    \sum_{b\in\mathcal{N}_{s,t}} \sum_{p\in\mathcal{P}_{s,t}} x_{b,v,p}.
\)
and we write $x_{v, p} = \sum_{b\in \mathcal N_{s, t}}x_{b, v, p}$.
As in Section~\ref{sec:structure}, we can view $X_{s,t} = (X_{(s,t), v})_{v\in\Supp}$ as the aggregate EA-LP allocation over values for type $(s,t)$. Let $w_{s,t}$ be the important value of this vector, defined exactly as in the single-minded case%
\footnote{A value exchange argument, similar to that used in the proof of Lemma~\ref{lem:structure-S}, can be applied on paths in order to characterize the structure of $X_{s,t}$'s entries and to establish the existence of an important value$w_{s,t}$  for each generalized bundle $B_{s,t}$.}. 
Depending on whether $w_{s,t}$ is tight or crucial, we apply the same three-case rule as in Section~\ref{sec:framework}, Step~2, to determine the number of copies of $B_{s,t}$ that we post at prices $w_{s,t}$ and $w_{s,t} + 1$. 

Formally, the menu for type $(s,t)$ consists of $c_{s,t}$ copies of $B_{s,t}$, each priced either at $w_{s,t}$ or $w_{s,t}+1$, according to the case distinction in Algorithm~\ref{alg:post-menu}. This fully specifies the number of copies offered at each price for each generalized bundle $B_{s,t}$.

We now define how each copy of $B_{s,t}$ is instantiated as a concrete $s$-$t$ path. For every copy of $B_{s,t}$, we independently sample a path $p \in \mathcal{P}_{s,t}$ using roulette drawing with probability proportional to EA-LP allocation (we omit the $(s, t)$ index, because the $(s, t)$ pair is fixed):

\[
    \Prob{\mbox{path $p$ is chosen}} = \frac{\sum_{v \ge w} x_{v,\,p}}{\sum_{p' \in \mathcal{P}_{s,t}} \sum_{v \ge w}
        x_{v,\,p'}}.
\]

All copies are sampled independently of each other. Once the menu is fixed, when
a buyer of type $(s,t)$ purchases a copy of $B_{s,t}$, they receive the pre-sampled
path associated with that copy.

This completes the definition of the menu and the routing in the online network routing setting.

\paragraph{Edge Load.} We now show that the expected load induced by this construction satisfies the same kind of bound as in Lemma~\ref{lem:L(e)}, yielding Property~2 for the load of network edges.

For a fixed type $(s,t)$ and edge $e$, let $L_{s,t}(e)$ denote the total load on edge $e$ due to buyers of type $(s,t)$ under the (unconstrained) execution of online algorithm with the path allocation menu defined above. Let $P_{s,t}^e \subseteq \mathcal{P}_{s,t}$ be the set of $s$-$t$ paths that use edge $e$. The following lemma can proved similarly to Lemma~\ref{lem:L(e)} (the proof is deferred to Appendix \ref{sec:graphload}). 

\begin{lemma}
\label{lem:graph-load}
For every edge $e\in E$, the path allocation of the unconstrained algorithm with the routing menu guided by the $\gamma$-scaled down EA-LP and constructed as above satisfies:
\[
    \E{M,I}{L_{s, t}(e)}
    \;\le\;
    \E{M,I}{Z_{s, t}(e)}
    \;\le\;
    \sum_{b \in \mathcal{N}_{s,t}} \sum_{v\in \Supp} \sum_{p\in P_{s,t}^e} x_{b,v,p}
    \;\le\;
    \frac{c(e)}{\gamma},
\]
where $Z_{s, t}(e)$ is a sum of negatively associated random variables.
\end{lemma}

Combining Lemma~\ref{lem:graph-load} with the expected social welfare guarantee for the unconstrained algorithm of Lemma~\ref{lem:ualg}, and using the value of $\gamma$ and the analysis used for the general single-minded setting in Section~\ref{sec:general-single-minded}, we obtain a competitive ratio of $O(m^{1/(B+1)})$ for online network routing (Theorem~\ref{thm:graph}).

\section{Lower Bounds}
\label{sec:lower_bounds}

Next, we prove Theorem~\ref{thm:lb-gen} and give a proofsketch of Theorem \ref{thm:lb-d}. The full proof of Theorem \ref{thm:lb-d} along with refined lower bounds for the case of unit capacities are presented in Appendix~\ref{app:lower-bounds}.

\subsection{The Proof of Theorem \ref{thm:lb-gen}}

It suffices to consider the case where $1\le B<\ln m$ (if $B\ge \ln m$, the trivial $\Omega(1)$ bound implies the claim). We use $r$-way qualitatively independent (QI) partitions introduced by \citet{renyi}.

\begin{definition}[$r$-way QI partitions]
Let $U$ be an $n$-element set. A partition of $U$ into $t$ classes is
$\mathcal{P}=\{P_1,\dots,P_t\}$, where the different classes are mutually disjoint and their union is $U$. A family
$
\mathcal{F}=\{\mathcal{P}^{(1)},\dots,\mathcal{P}^{(k)}\}
$
of $t$-partitions is \emph{$r$-way qualitatively independent} if for every choice of $r$
distinct partition indices $j_1,\dots,j_r\in[k]$ and $r$ class indices $a_1,\dots,a_r\in[t]$, one for each partition, it holds that 
$
P^{(j_1)}_{a_1}\cap\cdots\cap P^{(j_r)}_{a_r}\neq\emptyset.
$
Let $g(n,t,r)$ be the maximum size of such a family.
\end{definition}

\citet{poljak} proved the following.
\begin{theorem}[{\cite[Theorem~5]{poljak}}]\label{thm:poljak-again}
For all integers $n,t,r\ge 1$,
\[
g(n,t,r)\ \ge\ \frac{r}{et}\cdot \exp\!\left(\frac{n}{r t^r}\right).
\]
\end{theorem}

In the following, we set $n=m$ and $r=B+1$, and choose $t$ so that $g(m,t,B+1)\ge B\,t^t$. The proof is by direct substitution into Theorem \ref{thm:poljak-again} and is deferred to
Appendix~\ref{app:proofs}.

\begin{lemma}\label{lem:N-vs-tt}
For every $m\ge 16$ and $1\le B<\ln m$, define
\[
t=t_{m,B}\coloneqq \left\lfloor\left(\frac{m}{2B\ln m}\right)^{1/(B+2)}\right\rfloor.
\]
Then $g(m,t,B+1)\ge B\,t^{\,t}$.
\end{lemma}

\paragraph{Lower-Bound Instance.}
There are $m$ items and exactly $B$ copies of each item. Let $t=t_{m,B}$ and $N:=B\,t^t$.
By Lemma \ref{lem:N-vs-tt}, there exists a $(B+1)$-way QI family
$\{\mathcal{P}^{(1)},\dots,\mathcal{P}^{(N)}\}$ of $t$-partitions of $[m]$,
where $\mathcal{P}^{(i)}=\{P^{(i)}_1,\dots,P^{(i)}_t\}$.
For each $(i,a)\in[N]\times[t]$, we introduce a single-minded buyer $b_{i,a}$ who is interested in $P^{(i)}_a$ and has value $1$ with probability $1/t$ (and $0$ otherwise), independently over all $(i,a)$. We let $E_i$ denote the event that all $t$ buyers in group $i$ have value $1$.

The following two observations follow from the definition of the lower bound instance. 

\begin{observation}\label{ob:at-most-B-groups}
In any feasible allocation (respecting $B$ copies per item), it is impossible to serve buyers from more than $B$ distinct groups.
\end{observation}

\begin{observation}\label{ob:serve-B-groups}
If the events $E_{i_1},\dots,E_{i_B}$ occur for $B$ distinct groups, then there exists a feasible allocation that serves all $tB$ buyers in these groups.
\end{observation}

\paragraph{Analysis.}
We recall that $E_i$ denotes the event that all $t$ buyers in group $i$ have value $1$. Then, $\Prob{E_i}=p=(1/t)^t$. For $X:=\sum_{i=1}^N \ind{E_i}$, we have that $X\sim\mathrm{Bin}(N,p)$ and that $\mathbb{E}[X]=Np=B$. From the second case of the proof of Lemma~\ref{lem:PBD}, we get that $\Prob{X\ge B}\ge 0.5$. \qed

\paragraph{Offline versus Online.}
If $X\ge B$, then by Observation~\ref{ob:serve-B-groups}, the social welfare of the offline optimum is $B\,t$. Hence, the expected social welfare of the prophet is at least $0.5\,B\,t$.
By Observation \ref{ob:at-most-B-groups}, any online algorithm can collect value from at most $B$ groups. Specifically, for every group $i$, the number of buyers with positive value is $\mathrm{Bin}(t,1/t)$. Hence, once the algorithm commits to accepting buyers from some group, the expected remaining value to accept from that group is at most $1$.
Thus, $\mathbb{E}[\alg]\le 2B$, which implies a competitive ratio of at least $t/4$. 
For the proof of Theorem \ref{thm:lb-gen}, we use 
\[ t = \left\lfloor \left(\frac{m}{2B\ln m}\right)^{1/(B+2)}\right\rfloor \ge \frac{1}{2} \left(\frac{m}{2B\ln m}\right)^{1/(B+2)}\,. \]
Observing that $(2B)^{1/(B+2)}\le 2$, for all $B\ge 1$, we get that 
$4t\ge (m/\ln m)^{1/(B+2)}$. Therefore, we obtain the following lower bound on the competitive ratio of any online algorithm, thus concluding the proof Theorem \ref{thm:lb-gen}:
\[
\frac{1}{16}\left(\frac{m}{\ln m}\right)^{1/(B+2)} \,.
\]

\subsection{Proof Sketch of Theorem \ref{thm:lb-d}}

The proof is similar to the proof of Theorem \ref{thm:lb-gen}. The only new ingredient is enforcing $\abs{S_b}\le d$.
We do so by replacing Poljak's partitions with \emph{$d$-balanced} partitions (all classes have size at most
$d$), obtained by sampling balanced labelings (with each label appearing exactly $d=n/t$ times). Using negative association, we get the balanced analogue $g_{\mathrm{bal}}(n,t,r,d)\ \ge\ er/t\cdot \exp\!\left(\frac{n}{r t^r}\right)$. 

Setting $r=B+1$, $n=m=dt$ and
$
t=\lfloor\big(d/({2B\ln d})\big)^{1/(B+1)}\rfloor
$
gives a $(B+1)$-way QI family of size $N=B\,t^t$ of $t$-partitions of the $m$ items with class size at most $d$. The buyers are defined exactly as before (we have one buyer per class, their value is $1$ with probability $1/t$). So, every desired bundle has size at most $d$. The binomial argument and the two structural observations remain valid and imply a lower bound of $\Omega(t)=\Omega((d/\ln d)^{1/(B+1)})$ on the competitive ratio of any online algorithm. A detailed proof of Theorem \ref{thm:lb-d} can be found in Appendix~\ref{app:lower-bounds}. \qed

\section{Conclusion}

We introduced a unified approach for online resource allocation with strong complementarities via \emph{static} and \emph{anonymous} bundle-pricing. Our mechanisms are guided by an ex-ante LP relaxation with scaled down capacities, and translate the fractional solution’s threshold structure into posted bundle prices. This yields improved guarantees as multiplicity grows: an $O(d^{1/B})$-competitive mechanism for $d$-single-minded instances, and $O(m^{1/(B+1)})$-competitiveness for general single-minded combinatorial auctions and for online network routing. We also prove information-theoretic lower bounds via an embedding of the classical prophet-inequality hard instance using qualitative-independence constructions, showing that a similar dependence on $m$ (or $d$) and $B$ is essentially unavoidable. Two directions seem particularly interesting: (i)~closing the remaining gaps for both the general and the $d$-single-minded setting; and (ii)~studying sample-based implementations, where the value distribution of the buyers are implicitly inferred from a small (polynomially many to begin with, ideally a constant) number of samples. 


\bibliographystyle{plainnat}
\bibliography{main}

\newpage
\appendix

\section*{\LARGE\bf Appendix}
\medskip
\section{Different Distribution Supports} \label{app:support-assumption}

Our results extend to the setting where each buyer $b$ has a distribution $\mathcal D_b$ over its own support $\Supp_b$ and by a standard scaling argument we can assume that $\Supp_b \subseteq \mathbb{Z}_{\ge 0}$ without loss of generality. We extend each buyer-specific support $\Supp_b$ to a common support $\Supp$, run our $c$-approximation algorithm in the resulting same-support instance, and post the induced randomized menu. We show that with the appropriate extension, this (randomized) menu yields a $4c$-approximation for the original problem with different supports.

Let $v_{\max}^{(b)} = \max_{v\in \Supp_b} v$ and $v_{\max} = \max_{b} v_{\max}^{(b)}$. Define $\Supp = \{0,\ldots ,v_{\max}\}\supseteq\cup_b \Supp_b$. To reduce to the equal-support case, we modify each distribution $\mathcal D_b$ so that all of them share the same support~$\Supp$. For each buyer $b$ with $\Supp\setminus\Supp_b \neq \emptyset$, define the modified distribution $\mathcal D_b'$ by
\[
q'_{b,v^{(b)}_{\max}}
    \;=\;
    q_{b,v^{(b)}_{\max}} - \varepsilon \cdot |\Supp \setminus \Supp_b|,
\qquad
q'_{b,v} = \varepsilon
    \quad \forall v \in \Supp \setminus \Supp_b,
\]
and leave all other $q_{b,v}$ unchanged. The parameter $\varepsilon>0$ is chosen sufficiently small so that
\[
0 < q'_{b,v^{(b)}_{\max}}
\quad\text{i.e. }\quad
\varepsilon \le C_1 \coloneqq
    \min_b \frac{q_{b,v^{(b)}_{\max}}}{2\,|\Supp \setminus \Supp_b|}.
\]

The modified distribution $\mathcal D_b'$ is a valid distribution, and all $\mathcal D_b'$ share common support $\Supp$. Define $\mathcal{D} \coloneqq \mathcal{D}_1\times\ldots\times \mathcal{D}_n$ and $\mathcal{D'} \coloneqq \mathcal{D}'_1\times\ldots\times \mathcal{D}'_n$.

Let $\fopt$ and $\fopt'$ be the fractional LP optima under the original and modified distributions, respectively. Now take an optimal LP solution $x$ for $\mathcal D$ and decrease $x_{b,v^{(b)}_{\max}}$ by 
$$\min \left\{x_{b,v^{(b)}_{\max}}, \varepsilon|\Supp\setminus \Supp_b|\right\}$$ 
for each buyer $b$; the resulting solution is feasible for $\mathcal D'$ and loses at most $\varepsilon\abs{\Supp\setminus \Supp_b}\, v_{\max}^{(b)}$ in value for each buyer $b$. Thus,
\[
\fopt'
    \ge
    \fopt
    \;-\;
    \varepsilon \sum_b |\Supp\setminus \Supp_b|\, v^{(b)}_{\max}.
\]
Thus, by choosing
\[
\varepsilon
    \le C_2 \coloneqq
    \min\!\left\{
        C_1,\;
        \frac{\fopt}{2\sum_b |\Supp\setminus \Supp_b|\, v^{(b)}_{\max}}
    \right\},
\]
we guarantee $\fopt' \ge \frac12 \fopt$.

Let $I \sim \mathcal D$ and $I' \sim \mathcal D'$ denote value realizations under the
original and modified distributions. For the total variation distance between $\mathcal{D}_b, \mathcal{D}_b'$ for each buyer $b$ we get by definition that $d_{\mathrm{TV}}(\mathcal D_b,\mathcal D_b') \;=\; \varepsilon\,|\Supp\setminus \Supp_b|$. Thus, since $\mathcal{D}, \mathcal{D}'$ are product distributions we get 
\[
d_{\mathrm{TV}}(\mathcal D,\mathcal D')
\;\le\;
\sum_b d_{\mathrm{TV}}(\mathcal D_b,\mathcal D_b')
\;=\;
\varepsilon\sum_b |\Supp\setminus \Supp_b|.
\]
For any fixed menu $M$ and any realization $I$, the welfare of our algorithm is
bounded as
\[
0 \;\le\; \alg_{M,I} \;\le\; n \cdot v_{\max}
\;\eqqcolon\; V_{\max}^{\mathrm{tot}}.
\]
Therefore, for any fixed $M$, using a standard property of the total variation distance (see Proposition 4.5 in \cite{levin_peres_wilmer_markovmixing}), we get
\[
\Bigl|\E{I}{\alg_{M,I}} - \E{I'}{\alg_{M,I'}}\Bigr|
\;\le\;
d_{\mathrm{TV}}(\mathcal D,\mathcal D')\cdot V_{\max}^{\mathrm{tot}}
\;\le\;
\varepsilon\!\left(\sum_b |\Supp\setminus \Supp_b|\right)\! V_{\max}^{\mathrm{tot}}.
\]
Taking expectation over the random menu $M$ yields
\[
\Bigl|\E{M,I}{\alg_{M,I}} - \E{M,I'}{\alg_{M,I'}}\Bigr|
\;\le\;
\varepsilon\!\left(\sum_b |\Supp\setminus \Supp_b|\right)\! V_{\max}^{\mathrm{tot}}.
\]
Thus, by additionally choosing
\[
\varepsilon \;\le\; 
C_3 \coloneqq
\min \left\{C_2, \frac{\fopt}{4c\,\left(\sum_b |\Supp\setminus \Supp_b|\right)\! V_{\max}^{\mathrm{tot}}}\right\},
\]
we get
\[
\E{M,I}{\alg_{M,I}}
\;\ge\;
\E{M,I'}{\alg_{M,I'}}
\;-\;
\frac{\fopt}{4c}
\;\ge\;
\,\frac{\fopt'}{c}
\;-\;
\frac{\fopt}{4c}.
\]
Combining with $\fopt' \ge \frac12\fopt$ gives
\[
\E{M,I}{\alg_{M,I}}
\;\ge\;
\frac{1}{4c}\,\fopt.
\]
Hence, a $c$-approximation algorithm for the equal-support case yields a $4c$-approximation algorithm for the general case, thus we can simply run our algorithm on the modified distributions $\mathcal D_b'$ and post the resulting randomized menu, justifying our assumption.

\section{Price Subadditivity} \label{app:subadditivity}

We now relax the simplifying assumption that a buyer who desires a bundle $S$ is only interested in $S$ copies. In principle, a buyer could purchase several bundles $T_1,\dots,T_\ell$ from the posted menu such that $\bigcup_i T_i \supseteq S$. We show that after a simple normalization step applied to the fractional LP solution, no buyer can ever strictly prefer such a cover over purchasing $S$ directly. In other words, the posted prices are subadditive.

\paragraph{Pre-processing.}
After solving the scaled-down EA-LP of Step~1, in Section~\ref{sec:framework}, let $\mathcal H$ denote the set of feasible fractional solutions $x$ attaining value $\fopt_\gamma$. Among these, we select a canonical solution by solving the secondary LP
\[
    \min \;\Bigl(\textsc{FracWeight}(x) + \sum_{e\in\mathcal M} w_e(x)\Bigr)
    \qquad\text{s.t. } x\in \mathcal H,
\]
where
\[
    \textsc{FracWeight}(x) = \sum_b\sum_{v\in\Supp} x_{b,v},     
    \qquad
    w_e(x) = \sum_{b:e\in S_b}\sum_{v\in\Supp} x_{b,v}.
\]
We henceforth work only with optimal solutions of this normalization LP. We first prove the following lemma. 

\begin{lemma}\label{lem:important-subadd}
    Fix a desired set $S$ and a value $v$ such that $x_{S, v} < q_{S, v}$. 
    Suppose there exist different desired sets $T_1,\dots,T_\ell$ such that $\bigcup_{i=1}^\ell T_i \supseteq S$.
    If $x_{T_i, v_i}>0$ for all $i\in [\ell]$ and some $v_i$, then
    \[
        \sum_{i=1}^\ell v_i > v .
    \]
\end{lemma}
\begin{proof}
    Suppose, for the sake of contradiction, that $\sum_{i=1}^\ell v_i \le v$. First, if $v=0$ the inequality trivially holds, thus we can assume that $v > 0$. Then, if $\sum_{i=1}^\ell v_i < v$, since $v_1,\dots,v_\ell$ are crucial or tight values and $v$ is non-tight we can decrease $\epsilon>0$ mass (for sufficiently small $\epsilon$) from each $x_{b_i, v_i}>0$ and increase $x_{b_j, v}$ by $\epsilon$, for some choice of buyers $b_i\in \mathcal{N}_{T_i},\,b_j\in \mathcal{N}_S$ for $i\in[\ell]$. Consequently, we remain feasible while increasing the objective; contradiction.

    Thus suppose that $\sum_{i=1}^\ell v_i = v$. We will get contradiction due to the objective of the pre-processing step. Indeed we can do the same transfer as before and get, 
    \begin{itemize}
    \item If $\ell=1$ then $T_1\supsetneq S$, so the load $w_e(x)$ decreases by $\epsilon$ on every item $e\in T_1\setminus S$, while $\textsc{FracWeight}(x)$ is preserved. Thus the total objective decreases, contradiction.
    \item If $\ell>1$ then the value shifted from $\ell$ different buyers is consolidated into a single buyer, so $\textsc{FracWeight}(x)$ decreases by $(\ell-1)\epsilon$, while each $w_e(x)$ remains the same or decreases. Again, the total objective decreases, contradiction.
    \end{itemize}
\end{proof}

\paragraph{Subadditivity.}
Fix a buyer whose desired set is $S$. Suppose that in the posted menu they purchase bundles $T_1,\dots,T_\ell$ and 
\[
\bigcup_i T_i \supseteq S \quad \text{and} \quad \sum_{i=1}^\ell p(T_i) < p(S) \tag{$\perp$} \label{eq:TvsS}
\]

Note that since we always offer $\abs{\mathcal{N}_S}$ copies of set $S$ at some price, a copy of $S$ will always be available. Indeed, for the buyer that arrives first this trivially holds and since we will prove that the menu is subadditive they will not purchase bundles of any other desired set, thus the $\abs{\mathcal{N}_S}$ are enough for each desired set $S$ (they are dedicated to buyers that desire $S$). We could offer $n$ copies (where $n$ is the total number of buyers) and avoid the above argument but $\abs{\mathcal{N}_S}$ copies suffice.

First, if $p(S)=1$ then $p(T_i)=0$, contradiction. Now notice that $x_{T_i, p(T_i)} > 0$ and $x_{S, p(S)-1}$ is not tight by the definition of our pricing rules. Thus, we can apply Lemma \ref{lem:important-subadd} with $v_i = p(T_i)$ and $v = p(S)-1$ and get
\[
\sum_{i=1}^\ell p(T_i) > p(S)-1,
\]
that contradicts with \eqref{eq:TvsS}.


\paragraph{Tie-Breaking.} We now show that no tie-breaking is needed by a small price perturbation. More specifically, instead of posting price $p$ for a copy we post price $f(p) = (1-\delta )\cdot p$, for $\delta < 1/(2R)$, where $R = v_{\max}$. Indeed $p-1 < f(p) < p$ thus $p$ value buyers now gain positive utility by purchasing the copy, so no tie-breaking is needed, while buyers with value $<p$ still do not afford buying the copy. Note that subadditivity of the menu still holds since the $(1-\delta )$ factor cancels out.

\section{Analysis of the General Algorithmic Framework} \label{app:gen-frame}

\subsection{Proof of Lemma \ref{lem:ualg}}\label{proof:ualg}

\begin{proof}
Fix a desired bundle $S$ and write $w = w_S$ for its important value. Let $\fopt_S \;=\; \sum_{b \in \mathcal{N}_S} \sum_{v \in \Supp} x_{b,v}\,v$ be the contribution of buyers in $\mathcal{N}_S$ to the scaled-down fractional optimum~$\fopt_\gamma$, and $x_{S,v} \;=\; \sum_{b\in\mathcal{N}_S} x_{b,v}$,  $q_{S,v} \;=\; \sum_{b\in\mathcal{N}_S} q_{b,v}$. By Lemma~\ref{lem:structure-S}, the vector $X^{(S)} = (x_{S,v})_{v\in\Supp}$ consists of a prefix of zeros, possibly a single crucial value $w$, and then a suffix of tight entries $x_{S,v}=q_{S,v}$ for all $v > w$. We now analyze the expected value obtained by the unconstrained algorithm $\ualg$ from buyers in $\mathcal{N}_S$ under the menu of Step~2, in Section~\ref{sec:framework}.

\medskip
\noindent
\textbf{Case 1: \emph{S} has no crucial value.}
Then $x_{S,v} = 0$ for $v \le w$ and $x_{S,v} = q_{S,v}$ for all $v > w$. The menu offers $|\mathcal{N}_S|$ copies of $S$ at price $w+1$. Since there are at most $|\mathcal{N}_S|$ buyers wanting $S$, no buyer is ever blocked in $\ualg$, and buyers with $v \le w$ never buy at price $w+1$. Thus
\[
    \E{}{\ualg\text{ value from }S}
    \;=\; \sum_{b\in\mathcal{N}_S}\sum_{v > w} q_{b,v}\,v
    \;=\; \sum_{b\in\mathcal{N}_S}\sum_{v > w} x_{b,v}\,v
    \;=\; \fopt_S.
\]

\medskip
\noindent
\textbf{Case 2: \emph{S} has a crucial value \emph{w} and $x_{S,w}\,w \le \sum_{v > w} x_{S,v}\,v$.}
The menu is again $|\mathcal{N}_S|$ copies of $S$ at price $w+1$, so the same argument as in Case~1 gives
\[
    \E{}{\ualg\text{ value from }S}
    \;=\; \sum_{v > w} x_{S,v}\,v = Q_S.
\]
By assumption we get that $Q_S \ge \fopt_S/2$, hence
\[
    \E{}{\ualg\text{ value from }S}
     \;\ge\; \frac12\,\fopt_S.
\]

\medskip
\noindent
\textbf{Case 3: \emph{S} has a crucial value \emph{w} and $x_{S,w}\,w > \sum_{v > w} x_{S,v}\,v$.}
Now the mass at $w$ dominates.  Let $Y$ denote the number of buyers in $\mathcal{N}_S$ whose value equals $w$; then $Y$ is a sum of independent binary random variables with total mean $q_{S, w}$. We consider two subcases, depending on $x_{S,w}$.

\medskip
\noindent
\emph{\textbf{Case 3a:} $x_{S,w} \le 1$.}
The menu offers:
$|\mathcal{N}_S|$ copies of $S$ at price $w+1$, and one additional copy at price $w$ with probability $p \;\coloneqq\; \max\{x_{S,w},\, x_{S,w}/q_{S,w}\}$. We distinguish according to $q_{S,w}$:
\medskip

\begin{itemize}
    \item If $q_{S,w} \le 1$, then $x_{S,w}/q_{S,w} \ge x_{S,w}$, so $p = x_{S,w}/q_{S,w}$.  
  By Lemma~\ref{lem:PBD}, $\Prob{Y \ge 1} \ge q_{S,w}/2$.  
  Conditioned on including the $w$-priced copy and on $Y \ge 1$, $\ualg$ sells that copy at value $w$, hence the expected contribution of this copy is at least
  \[
    p \cdot \Prob{Y \ge 1} \cdot w
    \;\ge\;
    \frac{x_{S,w}}{q_{S,w}} \cdot \frac{q_{S,w}}{2} \cdot w
    \;=\;
    \frac12\,x_{S,w}\,w.
  \]
  \item If $q_{S,w} > 1$, then $x_{S,w}/q_{S,w} < x_{S,w}$, so $p = x_{S,w}$.  
  Lemma~\ref{lem:PBD} yields $\Prob{Y \ge 1} \ge 1/2$, and thus the expected contribution of the $w$-priced copy is at least
  \[
      p \cdot \Prob{Y \ge 1} \cdot w
      \;\ge\;
      x_{S,w} \cdot \frac12 \cdot w
      \;=\;
      \frac12\,x_{S,w}\,w.
  \]
\end{itemize}
In either subcase,
\[
    \E{}{\ualg\text{ value from }S}
    \;\ge\; \frac12\,x_{S,w}\,w
    \;>\; \frac14\,\fopt_S.
\]
where in the last inequality we use $x_{S,w}\,w > \sum_{v > w} x_{S,v}\,v$.

\medskip
\noindent
\emph{\textbf{Case 3b:} $x_{S,w} > 1$.}
The menu offers: $\lfloor x_{S,w} \rfloor$ copies of $S$ at price $w$, and $|\mathcal{N}_S|$ copies at price $w+1$. Lemma~\ref{lem:PBD} with $\mu = q_{S,w}$ gives
\[
    \Prob{Y \ge q_{S,w}/2} \;\ge\; \frac12,
\]
and therefore
\[
    \Prob{Y \ge x_{S,w}/2}
    \;\ge\;
    \Prob{Y \ge q_{S,w}/2}
    \;\ge\; \frac12.
\]
Whenever $Y \ge x_{S,w}/2$, the unconstrained algorithm sells at least $x_{S,w}/2$ of the $\lfloor x_{S,w} \rfloor$ copies at price $w$. Hence
\[
    \E{}{\ualg\text{ value from }S}
    \;\ge\;
    \frac12 \cdot \frac{x_{S,w}}{2} \cdot w
    \;=\;
    \frac14\,x_{S,w}\,w
    \;>\;
    \frac18\,\fopt_S.
\]
Combining all cases, for every desired set $S$ we have
\[
    \E{}{\ualg\text{ value from }S}
    \;\ge\; \frac{1}{8}\,\fopt_S.
\]
Summing over $S$ and using $\sum_S \fopt_S = \fopt_\gamma$ yields
\[
    \E{M,I}{\ualg}
    \;=\;
    \sum_S \E{}{\ualg\text{ value from }S}
    \;\ge\;
    \frac{1}{8}\,\fopt_\gamma,
\]
which is exactly the statement of Lemma~\ref{lem:ualg}.
\end{proof}

\subsection{Proof of Lemma \ref{lem:L(e)}}
\begin{proof}
Fix a desired set $S$. We analyze the contribution to the load on a fixed item $e\in S$ coming from buyers in $\mathcal{N}_S$. As before, we assume that buyers in $\mathcal{N}_S$ only purchase menu entries corresponding to the bundle $S$; the general case follows from price subadditivity (Appendix~\ref{app:subadditivity}).  

For $b\in\mathcal{N}_S$ let $L_b(e)$ be the indicator that buyer $b$ contributes to the load on item $e$, and set $L_S(e) = \sum_{b\in\mathcal{N}_S} L_b(e)$. We inspect each of the menu cases of Step~2, in Section~\ref{sec:framework}.

\textbf{Case 1:} $\abs{\mathcal{N}_S}$ copies at price $w+1$.
Here a buyer accepts if and only if its realized value is $v>w$.  Thus
\[
    L_S(e) = \sum_{b\in\mathcal{N}_S} \ind{v_b > w} .
\]
Taking expectations,
\[
    \E{M,I}{L_S(e)}
    = \sum_{b\in\mathcal{N}_S} \sum_{v>w} q_{b,v}
    \le \sum_{b\in\mathcal{N}_S} \sum_{v\in\Supp} x_{b,v}.
\]

\textbf{Case 2:} one copy at price $w$ w.p.\ $x_{S,w}/q_{S,w}$, and $\abs{\mathcal{N}_S}$ copies at $w+1$.
A buyer contributes to $L_S(e)$ either by
\emph{(i)} obtaining the (possibly offered) $w$-copy, or
\emph{(ii)} having value $v>w$.
Since at most one $w$-copy can be claimed,
\[
    L_S(e)
    \;\le\;
    \ind{\text{$w$-copy is offered}} \cdot \ind{\exists b\,:\, v_b = w}
    \;+\;
    \sum_{b\in\mathcal{N}_S} \ind{v_b>w}.
\]
Taking expectations,
\[
    \E{M,I}{L_S(e)}
    \le
    \frac{x_{S,w}}{q_{S,w}}\cdot q_{S,w}
    + \sum_{b\in\mathcal{N}_S} \sum_{v>w} q_{b,v}
    =
    x_{S,w} + \sum_{b\in\mathcal{N}_S} \sum_{v>w} q_{b,v}
    = \sum_{b\in\mathcal{N}_S} \sum_{v\in\Supp} x_{b,v}.
\]

\textbf{Case 3:} one copy at price $w$ w.p.\ $x_{S,w}$, and $\abs{\mathcal{N}_S}$ copies at $w+1$.
Exactly as above,
\[
    L_S(e) \le \ind{\text{$w$-copy offered}} 
              + \sum_{b\in\mathcal{N}_S} \ind{v_b>w},
\]
and hence
\[
    \E{M,I}{L_S(e)}
    \le x_{S,w} + \sum_{b\in\mathcal{N}_S} \sum_{v>w} q_{b,v}
    = \sum_{b\in\mathcal{N}_S} \sum_{v\in\Supp} x_{b,v}.
\]

\textbf{Case 4:} $\lfloor x_{S,w}\rfloor$ copies at price $w$ and $\abs{\mathcal{N}_S}$ copies at $w+1$.
In this case,
\[
    L_S(e)
    \le
    \lfloor x_{S,w}\rfloor
    + \sum_{b\in\mathcal{N}_S} \ind{v_b>w},
\]
and
\[
    \E{M,I}{L_S(e)}
    \le
    \lfloor x_{S,w}\rfloor
    + \sum_{b\in\mathcal{N}_S} \sum_{v>w} q_{b,v}
    \le
    \sum_{b\in\mathcal{N}_S} \sum_{v\in\Supp} x_{b,v}.
\]

Summing over all desired sets $S$ concludes the proof. Moreover, in each case $L_S(e)$ is equal or upper bounded by a sum of negatively associated (NA) $\{0,1\}$ random variables ($Z_i$ in the lemma statement). Indeed, for \underline{Cases $1, 3, 4$} the proof is straightforward (note that in these cases the random variables are actually independent). 

For \underline{Case 2} we will use standard properties of negative dependence (see Section 2 in \cite{DubhashiRanjan1996}), first set $A_b=\ind{v_b=w}$, $B_b=\ind{v_b>w}$, and $O=\ind{\text{$w$-copy offered}}$. Then $A_b+B_b\le 1$, so $\{A_b,B_b\}$ is NA. By independent union, $\{A_b,B_b\}_{b\in\mathcal{N}_S}$ is NA. Applying the non-decreasing map $(A_b)_b\mapsto \ind{\sum_b A_b\ge 1}$ gives
$\{\ind{\exists b:\, v_b=w}, \{B_b\}_b\}$ is NA. Since $O$ is independent, adding $O$ preserves NA. Applying the non-decreasing map $(e,f)\mapsto e\cdot f$ yields
\[
\Big\{
\ind{\text{$w$-copy offered}}\cdot \ind{\exists b:\, v_b=w},
\ \{\ind{v_b>w}\}_{b\in\mathcal{N}_S}
\Big\}\ \text{is NA, done.}
\]

Finally, $L_{M,I}(e)=\sum_S L_S(e)$ inherits the same property since random variables for different sets are independent. This enables the concentration bounds we use in the analysis and finishes the proof of the lemma.
\end{proof}

\section{Graph Setting}

\subsection{Load Condition (Proof of Lemma \ref{lem:graph-load})}\label{sec:graphload}
\begin{proof}

The proof is similar to the single-minded case. First, fix a desired $(s, t)$ pair. We analyze the contribution to the load on a fixed edge $e$ coming from buyers in $\mathcal{N}_{s, t}$. As before, we assume that buyers in $\mathcal{N}_{s, t}$ only purchase menu entries corresponding to the bundle $(s, t)$; the general case follows from price subadditivity (Appendix~\ref{app:graph-subadd}).  

For $b\in\mathcal{N}_{s, t}$ let $L_b(e)$ be the indicator that buyer $b$ contributes to the load on edge $e$, and set $L_{(s, t)}(e) = \sum_{b\in\mathcal{N}_{s, t}} L_b(e)$. Recall that $\mathcal{P}_{s,t}^e$ denotes the set of $(s, t)$ paths that contain edge $e$. Lastly, let $x_{(s, t), \ge v, p} = \sum_{b\in \mathcal{N}_{s, t}}\sum_{v' \ge v}x_{b, v', p}$ and $x_{(s, t), \ge v}=\sum_{b\in \mathcal{N}_{s, t}}\sum_{v' \ge v}\sum_{p\in \mathcal{P}_{s,t}}x_{b, v', p}$. We inspect each of the menu cases of Step~2, in Section~\ref{sec:framework}.

\textbf{Case 1:} $\abs{\mathcal{N}_{s, t}}$ copies at price $w+1$.
Here a buyer accepts if and only if its realized value is $v>w$.  Thus
\[
    L_{(s, t)}(e) = \sum_{b\in\mathcal{N}_{s, t}} \ind{v_b > w}\ind{\text{$b$ routed through a path in $\mathcal{P}_{s, t}^e$}}.
\]
Taking expectations,
\[
    \E{M,I}{L_{(s, t)}(e)}
    \le q_{(s, t),>w}\sum_{p\in \mathcal{P}_{s, t}^e}\frac{x_{(s, t), \ge w, p}}{x_{(s, t), \ge w}} 
    \le  \sum_{p\in \mathcal{P}_{s, t}^e}x_{(s, t), \ge w, p}
    \le \sum_{b\in\mathcal{N}_{s, t}} \sum_{p\in \mathcal{P}_{s, t}^e}\sum_{v\in\Supp} x_{b,v,p}
\]

\textbf{Case 2:} one copy at price $w$ w.p.\ $x_{(s, t),w}/q_{(s, t),w}$, and $\abs{\mathcal{N}_{s, t}}$ copies at $w+1$.
A buyer contributes to $L_{(s, t)}(e)$ either by
(i) obtaining the (possibly offered) $w$-copy, or
(ii) having value $v>w$.
Since at most one $w$-copy can be claimed,
\begin{align*}
    L_{(s, t)}(e)
    &\;\le\;
    \ind{\text{$w$-copy is offered}} \cdot \ind{\exists b\,:\, v_b = w}\cdot \ind{\text{copy routed through a path in $\mathcal{P}_{s, t}^e$}}
     \\ &\;+\;
    \sum_{b\in\mathcal{N}_{s, t}} \ind{v_b>w}\cdot \ind{\text{$b$ routed through a path in $\mathcal{P}_{s, t}^e$}}.
\end{align*}
Taking expectations,
\begin{align*}
    \E{M,I}{L_{(s, t)}(e)}
    &\le
    \frac{x_{(s, t),w}}{q_{(s, t),w}}\cdot q_{(s, t),w} \cdot \sum_{p\in \mathcal{P}_{s, t}^e}\frac{x_{(s, t), \ge w, p}}{x_{(s, t), \ge w}}
    + q_{(s, t),>w}\sum_{p\in \mathcal{P}_{s, t}^e}\frac{x_{(s, t), \ge w, p}}{x_{(s, t), \ge w}} \\
    &=
    \sum_{b\in\mathcal{N}_{s, t}} \sum_{p\in \mathcal{P}_{s, t}^e}\sum_{v\in\Supp} x_{b,v,p}
\end{align*}

\textbf{Case 3:} one copy at price $w$ w.p.\ $x_{(s, t),w}$, and $\abs{\mathcal{N}_{s, t}}$ copies at $w+1$.
Exactly as above,
\begin{align*}
    L_{(s, t)}(e) &\le \ind{\text{$w$-copy offered}}\cdot \ind{\text{copy routed through a path in $\mathcal{P}_{s, t}^e$}} \\ 
    &+ \sum_{b\in\mathcal{N}_{s, t}} \ind{v_b>w}\cdot \ind{\text{buyer routed through a path in $\mathcal{P}_{s, t}^e$}},
\end{align*}
and hence taking expectations
\begin{align*}
    \E{M,I}{L_{(s, t)}(e)}
    &\le
    x_{(s, t),w}\cdot \sum_{p\in \mathcal{P}_{s, t}^e}\frac{x_{(s, t), \ge w, p}}{x_{(s, t), \ge w}}
    + q_{(s, t),>w}\sum_{p\in \mathcal{P}_{s, t}^e}\frac{x_{(s, t), \ge w, p}}{x_{(s, t), \ge w}} \\
    &=
    \sum_{b\in\mathcal{N}_{s, t}} \sum_{p\in \mathcal{P}_{s, t}^e}\sum_{v\in\Supp} x_{b,v,p}
\end{align*}

\textbf{Case 4:} $\lfloor x_{(s, t),w}\rfloor$ copies at price $w$ and $\abs{\mathcal{N}_{s, t}}$ copies at $w+1$.
In this case,
\begin{align*}   
    L_{(s, t)}(e)
    &\le
    \sum_{j=1}^{\floor{x_{(s, t),w}}}\ind{\text{$j$-th copy routed through a path in $\mathcal{P}_{s, t}^e$}} \\ 
    &+ \sum_{b\in\mathcal{N}_{s, t}} \ind{v_b>w}\cdot \ind{\text{buyer routed through a path in $\mathcal{P}_{s, t}^e$}},
\end{align*}
and taking expectations
\[
    \E{M,I}{L_{(s, t)}(e)}
    \le
    \floor{x_{(s, t),w}}\cdot \sum_{p\in \mathcal{P}_{s, t}^e}\frac{x_{(s, t), \ge w, p}}{x_{(s, t), \ge w}}
    +q_{(s, t),>w}\sum_{p\in \mathcal{P}_{s, t}^e}\frac{x_{(s, t), \ge w, p}}{x_{(s, t), \ge w}}
    \leq \sum_{b\in\mathcal{N}_{s, t}} \sum_{p\in \mathcal{P}_{s, t}^e}\sum_{v\in\Supp} x_{b,v,p}
\]

Summing over all desired $(s, t)$ pairs concludes the proof. Moreover, in each case $L_{(s, t)}(e)$ is equal or upper bounded by a sum of independent or negatively associated (NA) $\{0,1\}$ random variables. Indeed, for \underline{Cases $1, 3, 4$} the proof is straightforward. 

For \underline{Case 2}, we proceed exactly as in the single-minded case and we get that
\[
\Big\{
\ind{\text{$w$-copy offered}}\cdot \ind{\exists b:\, v_b=w},
\ \{\ind{v_b>w}\}_{b\in\mathcal{N}_{s, t}}
\Big\}\ \text{is NA.}
\]
and since we can multiply each element with independent binary random variables and still preserve NA, we get the result. 

Finally, $L_{M,I}(e)=\sum_{(s, t)} L_{(s, t)}(e)$ inherits the same property since random variables for different sets are independent. This enables the concentration bounds we use in the analysis and finishes the proof of the lemma.
\end{proof}

\subsection{Subadditivity}
\label{app:graph-subadd}
We now relax the simplifying assumption that a buyer who desires an $(s, t)$-pair is only interested in $(s, t)$-path copies. The proof is similar to the single-minded case. Similarly, a buyer could purchase several paths $p_1,\ldots,p_\ell$, that correspond to $(s_1, t_1),\ldots, (s_\ell, t_\ell)$ types respectively, from the posted menu such that $\bigcup_i p_i \supseteq p \in \mathcal{P}_{s, t}$. We will show that after applying essentially the same simple normalization step to the fractional LP solution as in the single-minded case, no buyer can ever strictly prefer such a cover over purchasing an $(s,t)$-path directly. In other words, the posted prices are subadditive.

\paragraph{Pre-processing.}
For the graph setting, the objective of the pre-processing step becomes
\[
    \min \;\Bigl(\textsc{FracWeight}(x) + \sum_{e\in E} w_e(x)\Bigr)
    \qquad\text{s.t. } x\in \mathcal H,
\]
where
\[
    \textsc{FracWeight}(x) = \sum_b\sum_{v\in\Supp}\sum_{p\in\mathcal{P}_b} x_{b,v,p}
    \quad
    \text{and}
    \quad
    w_e(x) = \sum_{b}\sum_{v\in\Supp}\sum_{p\in\mathcal{P}_b\,:\,e\in p} x_{b,v,p}.
\]
We now state the equivalent of Lemma \ref{lem:important-subadd} for the graph setting.
\begin{lemma}\label{lem:important-subadd-graph}
    Fix a desired $(s, t)$-path such that $x_{(s, t), v} < q_{(s, t), v}$. 
    Suppose there exist different $(s_i, t_i)$-paths $p_1,\dots,p_\ell$ such that $\bigcup_{i=1}^\ell p_i \supseteq p$, for some $p\in \mathcal{P}_{s, t}$.
    If $x_{(s_i, t_i), v_i, p_i}>0$ for all $i\in [\ell]$, then
    \[
        \sum_{i=1}^\ell v_i > v .
    \]
\end{lemma}
The proof proceeds exactly as in the single-minded case, except that we transfer mass from paths rather than from sets.

\paragraph{Subadditivity.}
We slightly adjust our notation for prices, writing them as $r$ rather than $p$, so as to avoid confusion with paths. As in the single-minded setting, fix a buyer whose desired routing is an $S=(s, t)$-path. Suppose that in the posted menu they purchase paths $p_1,\dots,p_\ell$, that correspond to $T_i = (s_i, t_i)$ pairs, such that 
\[
\bigcup_i p_i \supseteq p \quad \text{and} \quad \sum_{i=1}^\ell r(T_i) < r(S) \tag{$\perp\perp$}
\]
for some $p\in \mathcal{P}_{s, t}$. Note that since we always offer $\abs{\mathcal{N}_{s, t}}$ copies of $(s, t)$-paths at some price, a copy of an $(s, t)$-path will always be available as in the single-minded setting.

First, if $r(S)=1$ then $r(T_i)=0$, contradiction. Now notice that $x_{T_i, r(T_i)} > 0$ and $x_{S, r(S)-1}$ is not tight by the definition of our pricing rules. Now note that we cannot directly apply Lemma \ref{lem:important-subadd-graph} as in the single-minded setting since $x_{T_i, r(T_i)} > 0$ does not directly imply $x_{T_i, r(T_i), p_i} > 0$. However, we can overcome this issue via a mass-transfer argument. 

Suppose that $x_{T_i, r(T_i), p_i}=0$. Then, since $p_i$ is realized it holds that $x_{T_i,\geq w_{T_i}, p_i} > 0$ by the definition of the roulette drawing (recall that $w_S$ denotes the important value of $S$). Thus, there exists a value $v$ such that $x_{T_i, v, p_i} > 0$. Also, since $r(T_i)$ is a price there exists a path $p_i'$ such that $x_{T_i, r(T_i), p_i'} > 0$. We now transfer sufficiently small mass from $x_{T_i, v, p_i} > 0$ to $x_{T_i, r(T_i), p_i} = 0$ and from $x_{T_i, r(T_i), p_i'} > 0$ to $x_{T_i, v, p_i'}$. Note that these transfers do not interfere with the feasibility, optimality or the preprocessing step, since they occur between same paths and we balance the value exchanges. Thus, we can assume that $x_{T_i, r(T_i), p_i} > 0$ for each $i\in [\ell]$. 


As a result, we can apply Lemma \ref{lem:important-subadd-graph} and get contradiction as in the single-minded case.

\section{Lower Bounds} \label{app:lower-bounds}
\subsection{Proof of Theorem \ref{thm:lb-d}}
We give a variant of the lower-bound template from Theorem \ref{thm:lb-gen} that additionally enforces a maximum desired bundle size $d$. First we introduce the notion of balanced $r$-way QI partitions.

\paragraph{Balanced \emph{r}-way QI Partitions.}
Let $U$ be an $n$-element ground set.
A partition of $U$ into $t$ classes is called \emph{$d$-balanced}
if every class has size at most $d$.
Let
$g_{\mathrm{bal}}(n,t,r,d)$ denote the largest size of an $r$-way QI family
of $d$-balanced $t$-partitions of an $n$-element set.

We now refine Poljak's probabilistic argument to the balanced setting.
The only difference is that we sample partitions from uniformly random
\emph{balanced labelings}.

\begin{lemma}[Balanced refinement of Theorem 4 in \cite{poljak}]\label{lem:poljak-balanced}
Let $n,t,r\ge 1$ be integers with $t\mid n$, and set $d := n/t$.
Then
\[
    g_{\mathrm{bal}}(n,t,r,d)
    \;\ge\;
    \frac{r}{et}\cdot
    \exp\!\left(\frac{n}{r t^r}\right),
\]
\end{lemma}

\begin{proof}
Let $U$ be an $n$-element set with $n=dt$.
We generate a random $d$-balanced $t$-partition by choosing a uniformly random
labeling $\sigma:U\to[t]$ such that each label appears exactly $d$ times, and
letting the classes be $\sigma^{-1}(1),\dots,\sigma^{-1}(t)$.
Sample $k$ such partitions independently.

Fix distinct indices $j_1,\dots,j_r\in[k]$ and classes $a_1,\dots,a_r\in[t]$.
Let
\[
    Z
    :=
    \Bigl|
    P^{(j_1)}_{a_1}\cap\cdots\cap P^{(j_r)}_{a_r}
    \Bigr|
    =
    \sum_{u\in U}
    \ind{u\in P^{(j_1)}_{a_1}\cap\cdots\cap P^{(j_r)}_{a_r}}.
\]
For each $u\in U$ and each $\ell\in[r]$,
\[
    \Prob{u\in P^{(j_\ell)}_{a_\ell}} = \frac{d}{n} = \frac1t,
\]
and by independence across the $r$ sampled partitions,
\[
    \Prob{u\in P^{(j_1)}_{a_1}\cap\cdots\cap P^{(j_r)}_{a_r}}
    = \frac{1}{t^r}.
\]
Hence
\[
    \E{}{Z} = \frac{n}{t^r}.
\]

Moreover, for a fixed partition $j$ and class $a$,
$P^{(j)}_a$ is a uniformly random subset of $U$ of size $d$. It is a well known fact that for a uniformly random size-$d$ subset, the membership indicators are negatively associated (sampling without replacement) therefore
\[
    \Prob{Z=0}
    \;\le\;
    \prod_{u\in U}\Prob{u\notin P^{(j_1)}_{a_1}\cap\cdots\cap P^{(j_r)}_{a_r}}
    =
    \left(1-\frac{1}{t^r}\right)^n
\]

Let $\mathrm{Bad}$ be the event that the $k$ sampled partitions are not
$r$-way QI. By a union bound over the $\binom{k}{r}$ choices of distinct
indices and the $t^r$ choices of classes,
\[
    \Prob{\mathrm{Bad}}
    \;\le\;
    \binom{k}{r}\, t^r \left(1-\frac{1}{t^r}\right)^n.
\]
Similarly to \cite{poljak}, choosing
\[
    k
    =
    \ceil{
    \frac{r}{et}\cdot
    \exp\!\left(\frac{n}{r t^r}\right)
    }
\]
ensures $\Prob{\mathrm{Bad}}<1$, hence proving the lemma. 
\end{proof}

We now set $r := B+1$, since $(B+1)$-way QI is exactly what rules out serving buyers from more than $B$ groups when there are $B$ copies of each item (as in Claim~\ref{ob:at-most-B-groups}). We will take $m=dt$ items and require that every class in our partitions has size at most $d$; with $t\mid m$, balanced sampling gives classes of size exactly $m/t=d$.

Define, for $d$ sufficiently large and $1\le B<\ln d$,
\begin{equation}\label{eq:t-dB}
    t = t_{d,B}
    \;\coloneqq\;
    \left\lfloor
        \left(
            \frac{d}{2B\ln d}
        \right)^{\!\frac{1}{B+1}}
    \right\rfloor,
    \qquad
    m := dt.
\end{equation}

\begin{lemma}\label{lem:g-vs-tt-dB}
With $t$ and $m$ as in \eqref{eq:t-dB},
\[
    g_{\mathrm{bal}}(m,t,B+1,d)
    \;\ge\;
    B\, t^{\,t},
\]
for all sufficiently large $d$.
\end{lemma}

\begin{proof}
Apply Lemma \ref{lem:poljak-balanced} with $n=m=dt$, $r=B+1$, and $d=n/t$:
\[
    g_{\mathrm{bal}}(m,t,B+1,d)
    \;\ge\;
    \frac{(B+1)}{et}\cdot
    \exp\!\left(
        \frac{m}{(B+1)t^{B+1}}
    \right)
    =
    \frac{(B+1)}{et}\cdot
    \exp\!\left(
        \frac{d}{(B+1)t^{B}}
    \right).
\]
By the definition of $t$,
\[
    t^{B+1}
    \;\le\;
    \frac{d}{2B\ln d}
    \quad\Rightarrow\quad
    \frac{d}{(B+1)t^{B}}
    \;\ge\;
    \Omega(t\ln d),
\]
Hence
\[
    g_{\mathrm{bal}}(m,t,B+1,d)
    \;\ge\;
    d^{\Omega(t)}.
\]
Since $t \le d$, we have $d^{\Omega(t)} \ge B t^t$ for all
sufficiently large $d$, proving the claim.
\end{proof}

\paragraph{Lower-Bound Instance with Bundle Cap \emph{d}.}
Let $m=dt$ items and suppose there are exactly $B$ copies of each item.
By Lemma \ref{lem:g-vs-tt-dB}, there exists a $(B+1)$-way QI family
\[
    \mathcal{F}
    = \bigl\{\mathcal{P}^{(1)},\dots,\mathcal{P}^{(N)}\bigr\},
    \qquad
    N := Bt^t,
\]
of $d$-balanced partitions of the $m$ items into $t$ classes each:
\[
    \mathcal{P}^{(i)}=\{P^{(i)}_1,\dots,P^{(i)}_t\},
    \qquad |P^{(i)}_a|\le d.
\]
For each $i\in[N]$ and $a\in[t]$, introduce a single-minded buyer $B_{i,a}$
who wants exactly $P^{(i)}_a$ and whose value is
\[
    v_{i,a} =
    \begin{cases}
        1 & \text{with probability } 1/t, \\
        0 & \text{with probability } 1 - 1/t,
    \end{cases}
\]
independently over all $(i,a)$.
Thus every desired bundle has size at most $d$.

\paragraph{Analysis.}
Define events $E_i$ and the random variable $X$ exactly as in
Theorem \ref{thm:lb-gen}. Then
$X\sim\mathrm{Bin}(N,(1/t)^t)$ and $Np=B$, so
\[
    \Prob{X\ge B} \ge 0.5
\]
by the proof of Lemma \ref{lem:PBD} (see Case 2).

The structural claims from Theorem \ref{thm:lb-gen} remain unchanged:
\begin{itemize}
    \item By $(B+1)$-way QI, no feasible allocation can serve buyers from more
    than $B$ distinct groups.
    \item If $E_{i_1},\dots,E_{i_B}$ occur, all $tB$ buyers in these $B$ groups
    can be served.
\end{itemize}
Therefore,
\[
    \E{}{\mathrm{OPT}}
    \;\ge\;
    0.5\,Bt,
    \qquad
    \E{}{\mathrm{ALG}}
    \;\le\;
    2B,
\]
and hence
\[
    \frac{\E{}{\mathrm{OPT}}}{\E{}{\mathrm{ALG}}}
    \;\ge\;
    0.25\, t.
\]

\paragraph{Competitive Ratio as a function of \emph{d} and \emph{B}}
From \eqref{eq:t-dB}, for all sufficiently large $d$,
\[
    t
    \;\ge\;
    \frac{1}{2}
    \left(
        \frac{d}{2B\ln d}
    \right)^{\!\frac{1}{B+1}},
\]
and thus every online algorithm has competitive ratio at least
\[
    \Omega\!\left(
        \left(
            \frac{d}{\ln d}
        \right)^{\!\frac{1}{B+1}}
    \right).
\]

\subsection{Refined lower bounds for the unit-capacity case}

Here we show how to remove the $\ln m$ loss in this special case by using Poljak's explicit bound (Theorem~3 in \cite{poljak}) and prove the following theorem.

\begin{theorem} \label{thm:lb-1}
For the single-minded setting with $m$ items and unit capacities,
\emph{every} online algorithm has competitive ratio $\Omega\left(m^{1/3}\right)$.
\end{theorem}

Poljak's theorem states that if $t$ is a prime power and $t^2 - t$ divides $n - t$, then
\[
    g(n,t,2)
    \;\ge\;
    t^{\frac{n-t}{t^2 - t}}.
\]
We restrict to the case where the number of items is a cube of a prime, $m = n = p^3$ for a sufficiently large prime $p$, and set $t \coloneqq p,\ r \coloneqq 2$. Since $t=p$ is a prime (hence a prime power) and $t^2 - t = p^2 - p$ divides $n - t = p^3 - p = p(p^2-1) = p(p-1)(p+1)$, the condition of Poljak's \cite{poljak} Theorem~3 holds, and we obtain
\[
    g(n,t,2)
    \;\ge\;
    t^{\frac{n-t}{t^2 - t}}
    \;=\;
    p^{\frac{p^3-p}{p^2-p}}
    =
    p^{\frac{p(p^2-1)}{p(p-1)}}
    =
    p^{p+1}
    =
    t^{t+1}.
\]
In particular, $g(m,t,2) \;\ge\; t^{t+1} \;\ge\; t^t$, thus we select $N=t^t$ (the number of partitions in our family).  

We now reuse \emph{exactly} the same construction and analysis as in the proof of Theorem~\ref{thm:lb-gen}, specialized to $B=1$ ($r=2$), with this choice of $m$ and~$t$. Lemma~\ref{lem:PBD} gives $\Prob{X \ge 1} \;\ge\; 0.5$, since $Np = B$.

The structural part is exactly as in the $B$-copies proof, with $B=1$ and $r=2$:
the analogue of Claim~\ref{ob:at-most-B-groups} (for $B=1$) says that no feasible allocation
can serve buyers from more than one group; the analogue of Claim~\ref{ob:serve-B-groups}
(with $B=1$) says that if $E_i$ holds, the offline optimum can serve all $t$ buyers in group $i$.
Thus, repeating the offline argument from Theorem~\ref{thm:lb-gen} with $B=1$,
\[
    \E{}{\mathrm{OPT}}
    \;\ge\;
    t \cdot \Prob{X \ge 1}
    \;\ge\;
    0.5\, t.
\]

Similarly, the Rubinstein-style online argument from the proof of
Theorem~\ref{thm:lb-gen} with $B=1$ shows that each group has total expected value $1$ and,
once the algorithm takes value from a group, the expected remaining value in that group
is at most 1; combined with the fact that at most one group can be partially served when $B=1$,
this yields
\[
    \E{}{\mathrm{ALG}} \;\le\; 2.
\]

Therefore, for $m = p^3$ and all sufficiently large primes $p$,
\[
    \frac{\E{}{\mathrm{OPT}}}{\E{}{\mathrm{ALG}}}
    \;\ge\;
    \frac{0.5\, t}{2}
    =
    \Omega(t)
    =
    \Omega(p)
    =
    \Omega(m^{1/3})
\]
which proves Theorem~\ref{thm:lb-1}.

Leveraging the recent work of Saxena--Velusamy--Weinberg~\cite{SVW22} and Alon-Alweiss \cite{AA20} we can actually improve slightly the above lower bound. 

\begin{theorem} \label{thm:lb-1-new}
For the single-minded setting with $m$ items and unit capacities,
\emph{every} online algorithm has competitive ratio 
$m^{\,1/3 + \Omega(1/\log\log m)}$.
\end{theorem}

We first note that one can obtain a lower bound of this form by translating
the lower-bound instance of Saxena--Velusamy--Weinberg~\cite{SVW22} for prophet inequalities over $q$ partition matroids into a $q$-single-minded instance
and counting the number of items.

Instead of going through this translation, we give a direct combinatorial proof,
by connecting the Alon--Alweiss measure (AAM) to qualitatively independent partitions
and then plugging the resulting partitions into the single-minded lower-bound
template used in the proof of Theorem~\ref{thm:lb-1}.

\paragraph{AAM $\Rightarrow$ QI partitions.}

Fix a prime number $p$. For an integer $\ell\ge 1$, a set
$\mathcal{C} \subseteq \mathbb{Z}_p^\ell$ is called a \emph{$\mathbb{Z}_p$-covering $(\ell,p)$-family} if for all distinct $u,v \in \mathcal{C}$ and all $s \in \mathbb{Z}_p$ there exists an index
$i \in [\ell]$ such that
\[
    u_i - v_i \equiv s \pmod p.
\]
For an integer $N \geq 1$, the Alon--Alweiss measure $\AAM(p,N)$ is the smallest $\ell$
for which there exists such a covering family $\mathcal{C}$ with $\abs{\mathcal{C}}\ge N$. Let $g(n,t,2)$ denote the largest size of a family of pairwise (i.e. $2$-way) qualitatively independent $t$-partitions of an $n$-element set.

\begin{lemma}\label{lem:AAM-to-QI-final}
For every prime $p$ and every integer $N\geq 1$,
\[
    g\bigl(p\cdot \AAM(p,N),\,p,\,2\bigr) \;\ge\; N.
\]
\end{lemma}

\begin{proof}
Let $\ell:=\AAM(p,N)$. By definition there exists a $\Z_p$-covering $(\ell,p)$-family
$\mathcal{C}\subseteq \Z_p^\ell$ with $|\mathcal{C}|\ge N$, i.e.\ for all distinct $u,v\in\mathcal{C}$ and all
$s\in\Z_p$ there is $i\in[\ell]$ such that $u_i-v_i\equiv s\pmod p$.

Set $X:=[\ell]\times \Z_p$, so $|X|=p\ell$. For each $u=(u_1,\dots,u_\ell)\in\mathcal{C}$ define
a $p$-partition $\mathcal{P}_u$ of $X$ by blocks
\[
B_u(a)\;:=\;\{(i,a+u_i)\ :\ i\in[\ell]\}\qquad (a\in\Z_p),
\]
where addition is mod $p$. Then $\{B_u(a):a\in\Z_p\}$ is a partition of $X$ since every
$(i,j)\in X$ lies in the unique block with $a\equiv j-u_i\pmod p$.

We claim that $\mathcal{P}_u$ and $\mathcal{P}_v$ are qualitatively independent for all distinct $u,v\in\mathcal{C}$.
Fix $a,b\in\Z_p$ and put $s:=b-a\in\Z_p$. By the covering property, there exists
$i\in[\ell]$ such that $u_i-v_i\equiv s\pmod p$. Then
\[
(i,a+u_i) \in B_u(a)
\quad\text{and}\quad
a+u_i \equiv a+v_i+s \equiv b+v_i \pmod p,
\]
hence $(i,a+u_i)=(i,b+v_i)\in B_v(b)$ and so $B_u(a)\cap B_v(b)\neq\emptyset$.
Since $a,b$ were arbitrary, $\mathcal{P}_u$ and $\mathcal{P}_v$ are qualitatively independent.

Thus $\{\mathcal{P}_u:u\in\mathcal{C}\}$ is a family of $|\mathcal{C}|\ge N$ pairwise qualitatively independent
$p$-partitions of an $n=p\ell=p\cdot \AAM(p,N)$ element set. Therefore
$g\bigl(p\cdot \AAM(p,N),\,p,\,2\bigr)\ge N$.
\end{proof}

Saxena--Velusamy--Weinberg~\cite{SVW22} show that there exists an absolute
constant $\alpha>0$ such that for all sufficiently large primes $p$,
\[
    \AAM(p,p^p) \;\le\; p^{\,2 - \alpha/\log\log p}.
\]
Applying Lemma~\ref{lem:AAM-to-QI-final} with $N=p^p$ and $t:=p$, we obtain
\[
    g\bigl(t\cdot \AAM(t,t^t),\,t,\,2\bigr) \;\ge\; t^t
\]
and, writing $m := t \cdot \AAM(t,t^t)$, we have
\[
    m \;\le\; t^{\,3 - \alpha/\log\log t}.
\]
This implies that for all sufficiently large such $m$ there exists an absolute constant $c>0$
with
\begin{equation}\label{eq:t-lower-final}
    t \;\ge\; m^{\,1/3 + c/\log\log m}.
\end{equation}

\paragraph{Lower Bound Instance.}

For such an $m$ and $t$ we now take a ground set of $m$ items and fix a family
\[
    \mathcal{F}
    = \bigl\{\mathcal{P}^{(1)},\dots,\mathcal{P}^{(N)}\bigr\},
    \qquad N = t^t,
\]
of $N$ pairwise qualitatively independent $t$-partitions of the items, guaranteed
by Lemma~\ref{lem:AAM-to-QI-final}.

We then instantiate \emph{exactly} the single-minded lower-bound construction
from the proof of Theorem~\ref{thm:lb-1} for $B=1$ using this family $\mathcal{F}$.

The structural and probabilistic analysis of this instance (number of “fully active”
groups, offline prophet value, and the Rubinstein-style bound on any online algorithm)
is identical to the $B=1$ case of Theorem~\ref{thm:lb-1} and implies:

\[
    \mathbb{E}[\mathrm{OPT}] \geq 0.5t
    \qquad\text{and}\qquad
    \mathbb{E}[\mathrm{ALG}] \leq 2,
\]
for every online algorithm. Thus
\[
    \frac{\mathbb{E}[\mathrm{OPT}]}{\mathbb{E}[\mathrm{ALG}]}
    = 0.25t
    \;=\;
    \Omega\!\left(m^{1/3 + c/\log\log m}\right)
\]
by~\eqref{eq:t-lower-final}, for all sufficiently large $m$ of the form
$m = t\cdot\AAM(t,t^t)$ with $t$ prime.


\section{Deferred Proofs} \label{app:proofs}
\subsection{Fractional Relaxation and Offline Optimum}

We prove the following lemma in order to relate the fractional relaxation and the offline optimum. 
\begin{lemma} \label{lem:fopt-opt}
   $\fopt \geq \opt$.
\end{lemma}

\begin{proof}
Consider the optimal allocation for a fixed realization of buyers' values.  
For each buyer $b$ and value $v \in \Supp$, let $y_{b,v} \in \{0,1\}$ indicate whether in this realization buyer $b$ is allocated bundle $S_b$ when their value is $v$.  
By feasibility of the offline solution, for every $e \in \mathcal{M}$ we have
$$\sum_{b : e \in S_b} \sum_{v \in \Supp} y_{b,v} \le c(e),$$ and moreover $y_{b,v} = 0$ unless buyer $b$ realizes value $v$, implying $\E{I}{y_{b,v}} \le q_{b,v}$, where the expectation is over the randomness in the buyers' realized values.
Define the fractional allocation $x_{b,v} := \E{I}{y_{b,v}}$. By construction, $0 \le x_{b,v} \le q_{b,v}$, so the box constraints of the LP are satisfied. Taking expectations of the offline capacity constraints yields, for every $e \in \mathcal{M}$,
$$
    \sum_{b : e \in S_b} \sum_{v \in \Supp} x_{b,v}
    = \E{I}{ \sum_{b : e \in S_b} \sum_{v \in \Supp} y_{b,v}}
    \le c(e),
$$
so the LP capacity constraints are also satisfied. The objective value of this feasible fractional solution is
$$
    \sum_{b \in [n]} \sum_{v \in \Supp} x_{b,v} \cdot v
    = \E{I}{ \sum_{b \in [n]} \sum_{v \in \Supp} y_{b,v} \cdot v }
    = \opt,
$$
where the last equality holds because the offline allocation attains the optimal welfare for each realization. Since $\fopt$ is the optimal value of the LP and we constructed a feasible fractional solution with value $\opt$, we obtain $\fopt \ge \opt$.
\end{proof}

\subsection{Proof of Lemma \ref{lem:structure-i}}
\begin{proof}
Fix a buyer $b$ and consider the vector $X^{(b)} = (x_{b,v})_{v \in \Supp}$, where values in $\Supp$ are ordered increasingly. Recall the value‐exchange argument from the main text: whenever $v' > v$, we may shift an infinitesimal amount of mass from $x_{b,v}$ to $x_{b,v'}$ without violating any LP constraints, unless one of the box constraints $0 \leq x_{b,v}$ or $x_{b,v'} \le q_{b,v'}$ becomes tight. Such a shift strictly increases the objective because $v' > v$.

\paragraph{Case 1: A crucial value $w$ exists.}

Suppose there exists a value $v < w$ with $x_{b,v} > 0$. Since $v < w$ and $x_{b,w} < q_{b,w}$, we may shift mass from $x_{b,v}$ to $x_{b,w}$ without violating feasibility, strictly improving the objective - a contradiction. Thus $x_{b,v}=0$ for all $v<w$. Next, suppose there exists a value $u > w$ with $x_{b,u} < q_{b,u}$. Since $x_{b,w} > 0$, we may shift mass from $x_{b,w}$ to $x_{b,u}$, again increasing the objective. Hence all $u>w$ must satisfy $x_{b,u} = q_{b,u}$.

\paragraph{Case 2: No crucial value exists.}
Here every $v$ satisfies either $x_{b,v}=0$ or $x_{b,v}=q_{b,v}$. Suppose there exist $u > v$ with $x_{b,u} = 0$ and $x_{b,v}=q_{b,v}$. Since $u>v$ and $x_{b,u}$ is not tight, shifting mass from $v$ to $u$ preserves feasibility and increases the objective - a contradiction. Thus all zero entries occur at smaller values, and all tight entries occur at larger values, with no interleaving.
\end{proof}

\subsection{Proof of Lemma \ref{lem:structure-S}}

\begin{proof}
Fix a bundle $S$ and let $X = X^{(S)}$ for simplicity. Recall that
\[
    x_{S,v} \;=\; \sum_{b \in \mathcal{N}_S} x_{b,v}
    \qquad\text{and}\qquad
    q_{S,v} \;=\; \sum_{b \in \mathcal{N}_S} q_{b,v},
\]
and we refer to $x_{S,v}$ as tight if $x_{S,v} = q_{S,v}$ and crucial if $0 < x_{S,v} < q_{S,v}$. We begin with two structural claims.

\begin{lemma}\label{lem:unique-crucial}
If any buyer $b \in \mathcal{N}_S$ has a crucial value, then all crucial values across buyers in $\mathcal{N}_S$ occur at the same value.  We call this value the \emph{crucial value of $X$}.
\end{lemma}

\begin{proof}
Suppose buyer $b$ has crucial value $w_b$ and buyer $b'$ has crucial value $w_{b'}$ with $w_{b} < w_{b'}$.  Because both buyers demand the same bundle $S$, shifting a small amount of mass from $x_{b,w_b}$ to $x_{b',w_{b'}}$ preserves all LP capacity constraints: both variables consume exactly the same set of items. This shift increases the objective because $w_{b'} > w_b$, contradicting optimality of the LP solution. Hence no two buyers in $\mathcal{N}_S$ can have crucial values at different levels, proving the claim.
\end{proof}

\begin{lemma}\label{lem:X-prefix-suffix}
If $X$ has a crucial value $w$, then
\[
    x_{S,v} = 0 \;\text{ for all } v < w,
    \qquad\text{and}\qquad
    x_{S,v} = q_{S,v} \;\text{ for all } v > w.
\]
\end{lemma}

\begin{proof}
Since $w$ is crucial for $X$, we have $0 < x_{S,w} < q_{S,w}$. First suppose there exists $v < w$ with $x_{S,v} > 0$. Then there is some buyer $b \in \mathcal{N}_S$ with $x_{b,v} > 0$. Because $x_{S,w} < q_{S,w}$, there exists some buyer $b' \in \mathcal{N}_S$ with $x_{b',w} < q_{b',w}$. 

By the same value–exchange argument used in Lemma~\ref{lem:structure-i}, we can shift an infinitesimal amount of mass from $x_{b,v}$ to $x_{b',w}$: this preserves all LP constraints, since both variables correspond to the same bundle $S$ and all item capacities are affected in exactly the same way. The objective strictly increases because $w > v$, contradicting optimality. Hence $x_{S,v} = 0$ for all $v < w$.

Now suppose there exists $u > w$ with $x_{S,u} < q_{S,u}$. Then there is a buyer $b' \in \mathcal{N}_S$ with $x_{b',u} < q_{b',u}$. Since $x_{S,w} > 0$, there exists a buyer $b \in \mathcal{N}_S$ with $x_{b,w} > 0$. Again, applying the value–exchange argument, we can shift an infinitesimal amount of mass from $x_{b,w}$ to $x_{b',u}$. Feasibility is preserved (both entries use the same bundle $S$), and the objective strictly increases because $u > w$, contradicting optimality. Thus no such $u$ can exist, and we must have $x_{S,u} = q_{S,u}$ for all $u > w$.
This proves the claim.
\end{proof}

We now conclude the proof of Lemma \ref{lem:structure-S}. If no crucial value exists in $X$, then for every $v$ we have either $x_{S,v}=0$ or $x_{S,v}=q_{S,v}$, and by the same exchange argument as above, zeros must occur only at lower values and tight values only at higher values. If a crucial value exists, then by Lemma \ref{lem:unique-crucial} it is unique, and by Lemma \ref{lem:X-prefix-suffix} $X$ consists of zeros below it, the crucial value, and tight entries above it. Thus in all cases the vector $X$ has exactly the form stated in the lemma.
\end{proof}

\subsection{Proof of Lemma \ref{lem:blocked-small}}

\begin{proof}
A buyer is rejected only if some item in her bundle is oversold.  Thus
\[
    |\textsc{Blocked}(M,I)|
    \;\le\;
    \sum_{e \in \mathcal{M}}
    \sum_{k = c(e)+1}^{\infty} \ind{L_{M,I}(e)=k}\, (k - c(e)).
\]
Taking expectations and using linearity,
\[
    \E{M,I}{|\textsc{Blocked}(M,I)|}
    \;\le\;
    \sum_{e \in \mathcal{M}}
    \sum_{k=c(e)+1}^{\infty} (k-c(e))\, \Prob{L_{M,I}(e) = k}.
\]

By Lemma \ref{lem:L(e)}, the random load $Z_{M,I}(e)$ is the sum of negatively
associated Bernoulli variables with total mean $\mu_e \le c(e)/\gamma$.  A standard
Chernoff bound for negatively associated variables gives
\[
    \Prob{L_{M,I}(e)\ge k}
    \;\le\;
    \Prob{Z_{M,I}(e)\ge k}
    \;\le\;
    \left(\frac{e\,\mu_e}{k}\right)^k
    \;\le\;
    \left(\frac{e\,c(e)}{\gamma\,k}\right)^k.
\]
Let $k=c(e)+j$ with $j\ge1$.  Using $(1+\tfrac{j}{c})^{c+j}\ge e^j$ and our choice
of $\gamma$, one obtains
\[
    \left(\frac{e\,c(e)}{\gamma(c(e)+j)}\right)^{c(e)+j}
    \;\le\;
    \left(\frac{e}{\gamma}\right)^{c(e)+j} e^{-j}
    \;\le\;
    \frac{1}{20m}\, e^{-j}.
\]
Summing over $j$ gives
\[
    \sum_{k=c(e)+1}^{\infty} (k-c(e))\,\Prob{L_{M,I}(e)=k}
    \;\le\;
    \frac{1}{20m}.
\]
Summing over all $m$ items yields the bound.
\end{proof}

\subsection{Proof of Lemma \ref{lem:N-vs-tt}}

\begin{proof}
    Since $B < \ln m$ and $m\ge 16$, we have
    \[
        \left( \frac{m}{2B\ln m}\right)^{1/(B+2)} > 1,
    \]
    and therefore
    \[
        1 \;\le\; t
        \;\le\;
        \left( \frac{m}{2B\ln m}\right)^{1/(B+2)}
        \;\le\;
        m^{1/(B+2)}
        \;\le\;
        m.
    \]
    Thus $t$ is a well-defined integer in $[m]$.
    Applying Theorem~\ref{thm:poljak-again} with $n=m$ and $r=B+1$ gives
    \[
        g(m,t,B+1)
        \;\ge\;
        \frac{(B+1)}{et} \exp\!\left(\frac{m}{(B+1) t^{B+1}}\right).
    \]
    By the definition of $t$,
    \[
        t^{B+2}
        \;\le\;
        \frac{m}{2B\ln m}.
    \]
    Hence
    \[
        t^{B+1} \;\le\; \frac{m}{2B\ln m} \cdot \frac{1}{t},
    \]
    and therefore
    \[
        \frac{m}{(B+1)t^{B+1}}
        \;\ge\;
        \frac{m}{(B+1)\,\frac{m}{2B\ln m}\,\frac{1}{t}}
        \;=\;
        \frac{2Bt\ln m}{B+1}.
    \]
    Plugging this into the bound on $g(m,t,B+1)$ yields
    \[
        g(m,t,B+1)
        \;\ge\;
        \frac{(B+1)}{et} \exp\!\left(\frac{2Bt\ln m}{B+1}\right)
        \;=\;
        \frac{(B+1)}{et} \, m^{2Bt/(B+1)}.
    \]
It remains to show that
\[
    \frac{(B+1)}{et}\, m^{\frac{2Bt}{B+1}} \;\ge\; B t^t .
\]
Since we already established that $t \le m^{1/(B+2)}$, we have $m \ge t^{B+2}$ and therefore
\begin{align*}
\frac{(B+1)}{et}\, m^{\frac{2Bt}{B+1}}
&\ge \frac{(B+1)}{et}\, \bigl(t^{B+2}\bigr)^{\frac{2Bt}{B+1}}
= \frac{(B+1)}{et}\, t^{\frac{2B(B+2)}{B+1}\,t} \\
&= B t^t \cdot \frac{(B+1)}{eB}\, t^{\left(\frac{2B(B+2)}{B+1}-1\right)t-1}.
\end{align*}
Using $\left(\frac{2B(B+2)}{B+1}-1\right)t-1 \ge 1$ (as $t\ge 1$), we get $g(m,t,B+1)\ge B t^t,$ as claimed.
\end{proof}

\subsection{Auxiliary Lemma}

\begin{lemma}\label{lem:PBD}
    For a random variable $X$ with mean $\mu$ that is the sum of some independent random variables on $\{0,1\}$ we get:
    $$\Prob{X \geq \mu/2}\geq 0.5\cdot \min\{1, \mu\}$$
\end{lemma}

\begin{proof}
Let $X=\sum_{i=1}^n X_i$, where the $X_i$ are independent $\{0,1\}$-valued
random variables with $\Prob{X_i=1}=p_i$. Then $\mu=\E{}{X}=\sum_{i=1}^n p_i$.
\paragraph{Case 1: $0<\mu\le 1$.}
Here $\mu/2\in(0,1/2]$, hence (since $X$ is integer-valued) $ \{X\ge \mu/2\}=\{X\ge 1\}$. Thus
\[
\Prob{X\ge \mu/2}
=\Prob{X\ge 1}
=1-\Prob{X=0}
=1-\prod_{i=1}^n(1-p_i) \geq 1-\Exp{-\sum_{i=1}^n p_i}=1-e^{-\mu}.
\]
using $1-x\le e^{-x}$ for $x\in[0,1]$. For $\mu\in[0,1]$, it holds that $1-e^{-\mu}\geq \mu/2$ therefore we get the lemma. 

\paragraph{Case 2: $\mu\ge 1$.}
It is a standard fact for Poisson binomial random variables
(i.e., sums of independent Bernoulli variables) that any median $m$ of $X$ satisfies $\lfloor \mu \rfloor \le m \le \lceil \mu \rceil$. In particular, $m\ge \lfloor \mu \rfloor$. Since $\mu\ge 1$ implies $\lfloor \mu \rfloor \ge \mu/2$, we obtain
\[
\Prob{X\ge \mu/2}\;\ge\;\Prob{X\ge m}\;\ge\;1/2.
\]
Combining the two cases, we conclude that
\[
\Prob{X\ge \mu/2}\ge 0.5\cdot\min\{1,\mu\}.
\]
\end{proof}

\end{document}